\documentclass{article}
\usepackage{amsmath}
\usepackage{amssymb}
\usepackage{slashbox}
\usepackage{graphicx}
\usepackage{verbatim}
\usepackage{lmodern}
\usepackage{array}
\usepackage{amsfonts}
\usepackage{latexsym}
\usepackage{amsthm}
\usepackage{graphicx}
\usepackage{calrsfs}
\usepackage{hyperref}
\usepackage{tikz}
\usepackage{eso-pic}
\usepackage{ifthen}
\usepackage{enumitem}
\usepackage{ccaption}
 
\usepackage[margin=0.8in]{geometry}
\usepackage{xr}
\newtheorem{Definition}{Definition}
\newtheorem{Addendum}{Addendum}[Definition]
\usepackage{mathtools}
\usepackage{amsmath}
\usepackage{amssymb}
\usepackage{amsthm}
\usepackage{mathrsfs}
\usepackage{graphicx}
\usepackage{dsfont}
\usepackage{wrapfig}
\usepackage{comment}
\usepackage{xcolor}
\usepackage{rotating}
\usepackage{hyperref} 
\usepackage[all]{xy}
\usepackage{caption}
\usepackage{subcaption}
\usepackage{mathtools}
\usepackage{fancyhdr}
\theoremstyle{remark}

 { \begin{list}%
         {--}%
         {\setlength{\labelwidth}{20pt}%
          \setlength{\leftmargin}{25pt}%
          \setlength{\topsep}{0pt}
          \setlength{\itemsep}{0pt}
          \setlength{\parsep}{0pt}}}%
 { \end{list} }

\newcommand{\R}{\mathbb{R}}

\newcommand{\I}{\mathcal{I}}

\newcommand{\Reeb}{\mathrm{R}}
\newcommand{\Ord}{\mathrm{Ord}}
\newcommand{\Ext}{\mathrm{Ext}}
\newcommand{\Rel}{\mathrm{Rel}}
\newcommand{\Dg}{\mathrm{Dg}\,}
\newcommand{\Mapper}{\mathrm{M}}

\newcommand{\im}{\mathrm{im}}

\newcommand{\Crit}{\mathrm{Crit}}
\newcommand{\Rips}{\mathrm{Rips}}
\newcommand{\DM}{\mathrm{DM}}
\newcommand{\End}{\mathrm{End}}

\newtheorem{thm}{Theorem}[section]

\newtheorem{defin}[thm]{Definition}

\newcommand{\e}{\varepsilon}

\newcommand{\dist}{\mathrm{d}}

\newcommand{\distgh}{\dist_{\rm GH}}

\makeatletter
\newcommand\Label[1]{&\refstepcounter{equation}(\theequation)\ltx@label{#1}&}
\makeatother
\newcommand{\beginsupplement}{%
        \setcounter{table}{0}
        \renewcommand{\thetable}{S\arabic{table}}%
        \setcounter{figure}{0}
        \renewcommand{\thefigure}{S\arabic{figure}}%
     }
\theoremstyle{plain}

\usepackage{transparent}
\begin{document}
  
\texttt{\textbf{\LARGE{Two-Tier Mapper: a user-independent clustering method for global gene expression analysis based on topology}}\\
$$Rachel  \ Jeitziner^a, Mathieu \ Carri\textrm{\`e}re^d,  Jacques \  Rougemont^b, $$
$$Steve \ Oudot^d, Kathryn  \ Hess^c, and  \ Cathrin  \  Brisken^a$$
$a$ Swiss Institute for Experimental Cancer Research,\\
$b$ Bioinformatics and Biostatistics Core facility,\\
$c$ Brain and Mind Institute,\\
School of Life Sciences, Ecole Polytechnique F\'ed\'erale de Lausanne, CH-1015 Lausanne, Switzerland\\
$d$ INRIA Saclay France\\}
There is a growing need for unbiased clustering methods, ideally automated. We have developed a topology-based analysis tool called Two-Tier Mapper (TTMap) to detect subgroups in global gene expression datasets and identify their distinguishing features. First, TTMap discerns and adjusts for highly variable features in the control group and identifies outliers. Second, the deviation of each test sample from the control group in a high-dimensional space is computed and the test samples are clustered in a global and local network using a new topological algorithm based on Mapper. Validation of TTMap on both synthetic and biological datasets shows that it outperforms current clustering methods in sensitivity and stability; clustering is not affected by removal of samples from the control group, choice of normalization nor subselection of data. There is no user induced bias because all parameters are data-driven. Datasets can readily be combined into one analysis. TTMap reveals hitherto undetected gene expression changes in mouse mammary glands related to hormonal changes during the estrous cycle. This illustrates the ability to extract information from highly variable biological samples and its potential for personalized medicine.

\section{Introduction}
Large datasets are generated at an exponentially increasing pace in biology and medicine, while the development of tools to analyze these data is lagging behind. The high variability of biological, in particular human, samples poses a challenge. It takes large sample numbers to understand the distribution of the data and to extract statistically significant features \cite{RNAseqNormalisation}. Often the choice of normalization is ambiguous and this affects the outcome of the analysis \cite{RNAseqNormalisation}.\\

Topology is a field of mathematics devoted to the study of shapes. Topological data analysis (TDA) is used to reduce dimensions and to recognize patterns (\cite{TD}, \cite{Chazal}). Global gene expression data samples for instance are considered as point clouds in a high-dimensional space. Topological methods can transform them into networks; the nodes are clusters of samples and the edges are determined by common samples between nodes \cite{Extracting}. Analysis of such networks enables discovery of specific patterns in any dataset. As topology is not sensitive to scale it is useful for highly variable biological data. \\

TDA approaches have been applied to numerous scientific domains including biology (\cite{Chazal}). A clustering method based on algebraic topology, Mapper, \cite{Extracting} has been applied to analyze large biological datasets, such as global gene expression profiles \cite{BC}, temporal single-cell RNA-seq data \cite{scRNAseqTOP}, and genomic data of viral evolution \cite{Viral}.

For the global gene expression analysis \cite{Monica2}, \cite{Monica3}, \cite{ReviewTop}, the data were pre-processed with a statistical tool and the combination of this statistical tool and Mapper is called Progression Analysis of Disease (PAD). Since the outcome of several statistical method, depending for the case of PAD on linear regression, can be strongly affected by outliers in both the control and the test group \cite{Outliers}, \cite{Handbookofclustering}, large sample numbers are critical for the method to render reliable results \cite{SampleSize}. Finding the outliers and removing them is troublesome in small datasets, where the definition of outliers is arbitrary \cite{SampleSize},\cite{Noiseestimation}. \\

Like other clustering methods such as $k$-means \cite{KMeans}, PAD, as well as Mapper alone, and hierarchical clustering depend on parameters the user choses; modifying parameters changes the output significantly \cite{StabilityKmeans}. Finally, clustering methods such as $k$-means do not verify stability results: small perturbations in the dataset can lead to different clusters and different conclusions \cite{Handbookofclustering}.\\

Here, we present a topology-based method inspired by PAD for global gene expression analysis particularly suited for small sample numbers ($n<25$), called Two-Tier Mapper (TTMap) which identifies significant variation and relatedness in datasets and i) can be used in a paired analysis, ii) takes into account batches, iii) is stable, and iv) does not require the user to choose any parameters. \\

\section{Results}
\subsection{Method description}
\subsubsection{Overview}

Each global gene expression profile represents a high dimensional vector in $\mathbb{R}^n$ with $n$ the number of genes. The input (Fig \ref{explanationfig} a, green) of Two-Tier Mapper (TTMap) is given by two matrices in log-2 scale, one for the control samples $\mathsf{N}$ the other for the test samples $\mathsf{T}$. Batches are defined as groups of samples distinguished by technical variation such as date and site of analysis, technical platform used or biological disparity such as different strains of mice.\\

TTMap comprises two independent parts, the Hyperrectangle Deviation Assessment (HDA) and the Global-to-Local Mapper (GtLMap). The first characterizes the control group and adjusts for outliers yielding the corrected control group that serves as a reference to calculate the deviation of each test vector individually. The second part uses the Mapper algorithm \cite{Mapper} where the parameters were carefully chosen; a two-tier cover, a special distance and an automated parameter of closeness. The two-tier cover detects global and local differences in the patterns of deviations thereby capturing the structure of the test group. The test samples are clustered according to the shape of their deviation. The extent of deviation of individual clusters translates into a color-code. A list of the differentially expressed genes is also provided (Fig \ref{explanationfig} a) (Details in \textit{Online Methods}).\\

\subsubsection{Hyperrectangle deviation assessment (HDA)}

Hyperrectangle deviation assessment (HDA) compares the value of each feature of any control sample $\mathsf N$ to the others in the same batch of the group $\mathsf N$ (Fig \ref{explanationfig} a, "adjustement of control group").  If the difference in absolute value is further than $e$, a parameter computed using the variances of all the genes (\textit{Online methods}), from the median of the others, it is considered an outlier and replaced by $\mbox{Not a Number (\textit{NA})}$. The numbers of replaced values in each sample in the control group (Fig \ref{explanationfig} a, shown by $\overline{N}_*$) are represented as a barplot (Fig \ref{explanationfig} b).  This allows the user to discern outlier samples for standard statistical analyses and to identify highly variable features of the control group (Fig \ref{explanationfig} b).


Thus, HDA creates a matrix that describes the range of expression values expected in group $\mathsf N$ corrected for outliers.  The $(k,j)$-coefficient of this matrix of the corrected control group, $(\overline{N}_k)_j$, which corresponds to the $j^{th}$ feature of sample $k$, is computed by:

\begin{equation*}
\begin{small}
((\overline{N})_k)_j = \begin{cases} NA &\mbox{if }   |(N_{k})_j -\textrm{median}_{i\in \mathcal I(N_k),i\neq k} (N_{i})_j|\geq e \\ 
(N_k)_j & \mbox{otherwise. }  \end{cases} 
\end{small},
\label{eq:01}\vspace*{-2pt}
\end{equation*}

\noindent Here, $(N_{i})_{j}$ denotes the value of the expression of gene $j$ in sample $i$, and $\mathcal{I}(N_{k}) \subseteq \{1, \dots,S\}$ is the set of indices of control samples in the batch containing $N_{k}$.
$\mbox{NA}$s are replaced by the median of the normal values in their batch.\\

Each feature has a range of values, in which control measurements are expected, for sample $T_{k}$ and gene $j$ given by $$B_j ^k= \big[\min_{i \in \mathcal{I}(T_{k})}({\overline{N}}_i)_j,\max_{i \in \mathcal{I}(T_{k})}({\overline{N}}_i)_j\big],$$
where  $\mathcal{I}(T_{k})$ is the set of indices of control samples in the batch containing $T_{k}$.
For each batch, these normal ranges determine a hyperrectangle in $n$-dimensional space $B_k = B_1^k \times \dots \times B_n^k$ (Fig \ref{explanationfig} c: example with $n=2$).\\

Each test sample $T_{{k}}$ is decomposed as $T_k = Nc.T_k + Dc.T_k$, where $Nc.T_k$ is the \emph{normal component}, which is its projection onto the hyperrectangle $B_k$ and hence is the closest point to $T_k$ inside $B_k$ (Fig \ref{explanationfig} c)  and the \emph{deviation component} $(Dc.T_{k})$, which is the remainder of the projection (Fig \ref{explanationfig} c) (\textit{Online Methods})  \\

More precisely, for each test sample $T_{{k}}$ and feature $j$, HDA computes
$$\bar{x}_j^k \in \big[\min_{i \in \mathcal{I}(T_{k})}({\overline{N}}_i)_j,\max_{i \in \mathcal{I}(T_{k})}({\overline{N}}_i)_j\big],$$
such that 
$$ | (T_{k})_j - \bar{x}_j^k|\leq | (T_{k})_j - x| $$ for all 
$$x \in \big[\min_{i \in \mathcal{I}(T_{k})}({\overline{N}}_i)_j,\max_{i \in \mathcal{I}(T_{k})}({\overline{N}}_i)_j\big].$$
 
Then,
\begin{equation*}
(Nc.T_{k})_j=  \bar{x}_j^k\quad \textrm{for all  }1\leq j\leq n \label{eq:01}\vspace*{-7pt} 
\end{equation*}
and
\begin{equation*}
(Dc.T_{k})_j=  (T_{k})_j - (Nc.T_{k})_j\quad\textrm{for all }1\leq j\leq n. \label{eq:01}\vspace*{-7pt}
\end{equation*}
\medskip
\subsubsection{Global-to-Local Mapper (GLMap)}
\label{GLMap}

The second step of TTMap first calculates distances and provides a visualization of these distances and relations in the dataset, in a manner analogous to Mapper \cite{Extracting}. It forms bins according to a measure of similarity on the test vectors.\\

The default similarity measure in GLMap is the \emph{mismatch distance}, $d_M$ given by a sum of mismatches, where a mismatch is defined by a gene that is differentially expressed in opposite direction as measured by the deviation component (\textit{Online Methods}, Fig \ref{explanationfig} d, n=1). The deviation must be bigger than $\alpha$ to avoid counting noise as mismatch. 
The mismatch distance, or sum of mismatches is defined as follows (Fig \ref{explanationfig} d), for a fixed $\alpha \geq 0$
\begin{equation*}
\begin{small}
d_M(X,Y) = \sum_{i=1}^n d_m((Dc.X)_i,(Dc.Y)_i),\\
\mbox{ where}
\end{small}
\end{equation*}
\begin{equation*}
\begin{small}
d_{m} (x,y) = \begin{cases} 0 &\mbox{if }  \textrm{sign}(x)=\textrm{sign}(y),\\ 
1 & \mbox{if }  \textrm{sign}(x) \neq\textrm{sign}(y) \\  
&\textrm{ and } |x| \textrm{ or }|y| \geq \alpha  \\
\frac{|x-y|}{8\alpha n} & \mbox{otherwise } \end{cases} 
\end{small}.
\end{equation*}

If features measured are gene expression values, then the default value does not need to be changed and is set to $\alpha= 1$, corresponding to a 2-fold-change, which is a standard cut-off for gene expression. \\

Furthermore, GLMap uses a \emph{filter function}, given by properties of interest of the samples. It can be chosen by the user to take into account relevant variables, such as the age of the patients in a cohort. The default filter function in GLMap, called \emph{total absolute deviation} and denoted $\tau$, measures the overall deviation of a test vector from the control, i.e., 
$$\tau : \mathsf T \rightarrow \mathbb{R} : T_k \mapsto \sum_{l \in S} \mid (Dc.T_k)_l \mid,$$
where $S$ is a subset of features, determined by the user, the default being to select all features, and $\mathsf T$ is the set of test vectors, which is a subset of $\mathbb R^{n}$.

Let $\operatorname{Im} \tau $ denote the \emph{image} of $\tau$ \textit{with multiplicity}, i.e., 
$$\operatorname{Im} \tau =\{ (\tau (X) , \sigma) \mid X\in \mathsf T, \ \sigma \in \{1, \dots, \textrm{mult}(X)\}\} \subseteq \mathbb R \times \mathbb N,$$

with the lexicographic order,
where $\textrm{mult}(X)= \textrm{card}(\tau^{-1}(\tau(X)))$ is the multiplicity of $\tau(X)$
and for any $0\leq a<b\leq 100$, let 
$$\tiny{q_{[a,b[} = \pi_1\big(\big\{y \in \operatorname{Im}\tau  \mid \textrm{quantile}_a(\operatorname{Im} \tau) \leq y < \textrm{quantile}_b(\operatorname{Im} \tau ) \big\}\big),}$$

where $\pi_1$ is the natural projection on the first component, and  $\textrm{quantile}_a(\operatorname{Im} \tau)$ is the a-th quantile of the ordered values in $\operatorname{Im} \tau$. \\

In default mode, GLMap applies the Mapper algorithm \cite{Extracting} to the quadruple given by the mismatch distance $d_{\text{M}}$, a closeness parameter $\epsilon$ (computed from the data ,\textit{Online Methods}, which depends on the variance in the control group), the total absolute deviation $\tau$, and the covering  of $\operatorname{Im} \tau$ given by 
$${\mathfrak{I} = \{ \operatorname{Im} \tau,  q_{[0,25[}, q_{[25,50[}, q_{[50,75[},  q_{[75,100]}\}.}$$
This means that GLMap performs single-linkage clustering with parameter $d_M$, i.e. two samples $X$ and $Y$ are clustered together if and only if there is a list of samples $X=X_{0}, X_{1},..., X_{n}=Y$ such that $d_{\text{M}}(X_{i}, X_{i+1})<\epsilon$ for all $0\leq i \leq n-1$ to 
\begin{itemize}
\item all of $\mathsf T$, giving the connected components $\{C_{01}, \dots, C_{0l(0)}\}$  of the graph $G_\epsilon$ defined by the vertex set $\{T_k\}$ and the edge set $\{(T_a,T_b) \textrm{ s.t. } d_M(T_a,T_b) < \epsilon \}$ and then to
\item the pre-image with respect to $\tau$ of each of the quantiles $q_{0,25}, q_{25,50}, q_{50,75},$ and   $q_{75,100}$, which gives the connected components $\{C_{i1}, \dots, C_{il(i)}\}$ of the subgraph $G_\epsilon(i) = \tau^{-1}(I_i)$, where $I_i \in \mathfrak{I}.$
\end{itemize}  
Two connected components $C_{ij}$ and $C_{kl}$ are represented as spheres with diameters increasing with the number of samples in each component. The spheres are connected by an edge whenever $C_{ij} \cap C_{kl} \neq \emptyset$, i.e. the algorithm links clusters that share samples as every sample is assessed twice for connectivity, once globally and once within its quartile, links are formed between local and global structures, enabling the discovery of subgroups based on the filter function of the global clusters (Fig \ref{explanationfig} a, Part2).\\

 The color of a sphere in the output figure of the method (see example in section \ref{flyatlas}, Figure \ref{figure3} a) is determined by the average of the values of the filter function applied to the samples in the bin. A legend for the color code is provided at the bottom of the output figure, for the size of the balls on the right, and for the different tiers on the left, i.e. the overall clustering and the clustering in the different quartiles, (Fig \ref{explanationfig} a, Part2). A list of the differentially expressed genes per cluster is provided.

\subsection{Theoretical aspects}
To assess the theoretical stability of TTMap, the effect of modifications of the source space, of the filter function and of approximations with a point cloud on its outputs was studied (\textit{Online methods}). Since there is no natural distance on the outputs of TTMap, one can not assess the stability directly on the TTMap graphs. Therefore, the information contained in the TTMap graphs is summarized as a diagram in $\mathbb{R}^2$ (Supplementary Fig \ref{fig:sign} d), similar to a persistence diagram (PD) \cite{Pers}, where there is a natural distance $d$ that generalizes the distance on PD, allowing a comparison of TTMap graphs.
 
The PD are summaries of the topological features of the graph (connected component, hole, branch, etc...) depicted as dots. Here, we supplemented PD with links between the ``local" features and the connected components (or the global clusters), forming a descriptor, denoted $DM(X,f, \mathfrak{I})$, for a space $X$ and a filter function $f: X \rightarrow \mathbb{R}$ that verifies mild regularity conditions.
In terms of these enriched PD, we establish the following theorems, stated informally here and precisely in the \textit{Online Methods} in Theorem \ref{completeness},\ref{thm:perturb_bis},\ref{thm:DStab_bis},  \ref{th:sig-approx} respectively.

 \begin{itemize}[leftmargin=*,align=left]
 \item  \textit{Completeness}
 The descriptor is complete, i.e., from the diagram $DM(X,f,\mathfrak{I})$ the information contained in the graph of $TTMap(X,f, \mathfrak{I})$ can be recovered. 
\item \textit{Stability with respect to changes of the filter function}
 If the filter function $f$ on the space is perturbed, the distance between the diagrams of $f$ and of its perturbation is not greater than the amount of perturbation.

\item \textit{Stability with respect to perturbations of the domain}
If the starting space $X$ is perturbed, then the distance between the diagrams of $X$ and of its perturbation depends linearly on the amount of perturbation. 
\item \textit{Stability with respect to point cloud approximations}

If data points are sampled on a space $X$, then the difference between the diagrams associated to $X$ and to the $\delta$-neighborhood graph built on the point cloud is less than a value depending on $\delta$.
 \end{itemize}

 \subsection{In silico validation}
 TTMap was tested on simulated data that mimics a situation for which standard methods are weak, i.e., small sample size (n$<$20). Moreover, differences in the subgroups arise from the same genes deviating in opposite directions. Control samples  $C_1, \dots C_6$ and test samples are generated composed of two subgroups $TA$ and $TB$, given by $TA_1 , TA_2, TA_3, TB_1, TB_2, TB_3,$ each with 10,000 features. The subgroups $TA$ and $TB$ have the same mean per gene as the mean of the control group, except for $m$ genes for which the mean is $\Delta$ times higher for TA, respectively lower for TB. The $m$ genes are true positives, whereas all the other features are true negatives. The accuracy of the method is estimated by simulating at least 30 datasets per condition and calculating the percentage of times it finds the right subgroups, establishing the clustering power of this method. Since TTMap is an analytical workflow we also assessed its performance in finding the genes that are differentially expressed. \\

 \subsubsection{TTMap's performance as a clustering method}
The performance of TTMap was assessed, with the parameter $\epsilon$ given by the lowest 2.5 percentile (Fig \ref{figure2} a) or the highest 2.5 percentile of the distribution of the distance $d_M$ between two random variables (Fig \ref{figure2} b) {with the} variance $\sigma^2$ ranging from $0.01$  to $1$ in order to measure the accuracy of TTMap in situation ranging from low variance to high variance. The number of significant features $m$ in the test cases were 50, 100, 500, 1000, and 5000, i.e., 0.5, 1, 5, 10, and $50\%$ of all the features, respectively. When $\Delta=2$, TTMap performed 100 \% correctly when the variance in the control group was in the biologically relevant range\cite{Noiseestimation}  (Fig \ref{figure2} a, b, pink shade), where $\sigma^2< 0.3$ (Fig \ref{figure2} a).  {For variances between $ 0.4$ and $0.8$ and for 0.5\% and 1\% of significant features respectively,} the method could no longer distinguish between noise and signal ($\Delta=2$) and classified all the samples as different. When $\epsilon$ is chosen in the higher 2.5 percentile (Fig \ref{figure2} b), the method was less good than the lower 2.5 percentile when the variances are low (below 0.5), but much better for higher variances (greater than 0.5). Moreover, the higher the number of significant features, the better TTMap performs in finding the two subgroups. Performance also improved when $\Delta$ increased (Fig \ref{figure2} a, Supplementary Fig \ref{SUP} a).\\
In contrast, a standard clustering tool \textit{Mclust }\cite{Mclust} that like TTMap does not need any parameter selection, was unable to find the right groups (Fig \ref{figure2} a, black line). This is in line with the fact that Mclust learns from the data, and hence requires a bigger sample size to be able to perform properly. Moreover, on this dataset the running time of Mclust is 45 times longer than that of TTMap (3.8 minutes versus 5 seconds, respectively). 
To assess whether the accuracy of TTMap relies solely on HDA or on GLMap, we applied Mclust to the data obtained after HDA, i.e. the deviation components. The accuracy of Mclust in detecting the subgroups  improved from 0 \% to 20\% on average  (Fig \ref{figure2} c). Thus, the accuracy of Mclust improved but did not reach the level of accuracy of TTMap.\\
  \subsubsection{TTMap's performance as a differential expression method in finding true positives and true negatives}
       
 To assess the performance of TTMap with regards to the genes determining a cluster, the numbers of true positives and of true negatives were computed. In datasets with low variance ($\sigma^2<0.5$) in the control group, TTMap found close to 100\% of the true positives and true negatives (Fig \ref{figure2} d and e). Since the samples in TA and TB have the same differentially expressed features but expressed in opposite directions, the moderated $t$-test did not detect any true positives. Even when the right groups are provided it poorly discovered the true positives in the subgroups, due to the low sample size (Fig \ref{figure2} f). Together with the observation that the moderated t-test finds close to 100 \% of true negatives, this suggests that the standard method is more likely to detect no significant genes in such a situation, and is therefore dominated by TTMap.\\
 
   \subsubsection{TTMap's performance on different sample sizes}
   \label{samplesize}
   TTMap was assessed on bigger datasets as well consisting of 100 or 200 simulated samples. The method performed as well at finding the right subgroups as in the case of small datasets (Supplementary Fig \ref{SUP} c). In particular, for small variances ($\sigma^2$ = 0-0.3) the method's accuracy is above 98\%, though it decreases for higher variances. 
   Different sizes of subgroups TA and TB were generated, i.e. two samples vs. four and one vs. five respectively. Even if one of the subgroups is composed only of a single sample, the method accurately (more than 98\% of accuracy for small variances) distinguishes it from the rest of the samples (Supplementary Fig \ref{SUP} b).

\subsection{TTMap characterize gene expression deviations of organs from whole fly tissues.}
\label{flyatlas}
To validate TTMap on a biological dataset, we analyzed the fly atlas (www.flyatlas.org). This dataset comprises 4 RNA replicate samples from $33$ drosophila tissues pooled from $50$ males and $50$ females (Supplementary Table \ref{Tab:01}) or third instar feeding larvae or wandering larvae . Global gene expression {of four} replicates from each tissue and of four replicates of whole flies were assessed. The group $\mathsf{N}$ to which each tissue was compared was {composed of} the "whole adult fly" samples. The number of expected subgroups corresponds to the number of organs.  \\

\subsubsection{TTMap compared to standard frameworks on real data}

To compare TTMap to established clustering methods, we used it in parallel with $k$-means \cite{KMeans} and DBSCAN \cite{DBSCAN}, to compute for how many organs the four replicates cluster uniquely together. The parameters for the two standard methods were chosen {to maximize} performance, i.e., $k$ in $k$-means was chosen to be equal to 33 (as there are 33 organs) and minPts in DBSCAN was {set} to 4, since there are four replicates. The $epsilon$ parameter of DBSCAN X was chosen according to guidelines in \cite{DBSCAN}. While DBSCAN and $k$-means clustered the four replicates of 20 and respectively 15 organs uniquely, {TTMap}, not provided with any parameter, clustered 21 organs uniquely (Fig \ref{figure3} b). 
To test the practical stability of TTMap, and compare it to DBSCAN and k-means, the data was quantile-normalised, only DBSCAN and TTMap exhibit stable performance in detecting uniquely clustering organs (Fig \ref{figure3} b), also reflected by the Rand Index (RI), a measure of similarity between two clusterings, which was 0.990 and 0.999 respectively. To further challenge the methods by randomly selecting 50\% of the genes and observe how the clustering is affected, DBSCAN performance drops from 20 to 8 uniquely clustering organs (RI=0.86) whereas TTMap remained stable, with 20 uniquely clustering organs (RI=0.995) (Fig \ref{figure3} c). Thus, TTMap is stable both upon normalization and random subselection. 

\subsubsection{The visual interpretation of TTMap}
TTMap computed that the organ that deviates the least from the whole adult fly (the control) is the whole larva (F) (Fig \ref{figure3} a). The two organs that deviate the most are testes (T) and brain (B) (Fig \ref{figure3} a). Surprisingly, one out of four spermatacea (K3) replicates clustered with three replicates of the adult thoracic muscle (V) and vice versa. This might explain the missed genes for K3 by standard tools and points to a potential labelling mistake (Supplementary Fig \ref{fig_fly_sup}). Replicates of the fatbody of the wandering larva (Wq) and the feeding larvae (Fq) clustered together globally (Fig \ref{figure3} a, overall). However, while three out of four feeding larvae (Fq) samples clustered in the 3rd quartile, the wandering larva samples were in the lowest quartile (Fig \ref{figure3} a). This indicates that the fatbody of the feeding larvae and the wandering larva share differentially expressed genes in comparison to the whole adult fly but these genes deviate to different extent from the control.\\

\subsection{Estrous cycle related gene expression changes in the mammary glands of C57-BL6 and Balb-C mice}
\label{Antoinesdata}

Next we challenged the method by asking whether TTMap can identify subtle gene expression changes that occur in an intact organ related to the alterations of hormone levels. For this, we recurred to an RNAseq dataset collected from intact mammary glands from C57/BL-6 and Balb-C females, which were staged to different phases of the estrous cycle (EC), proestrous (P), estrous (E), and diestrous (D) based on the prevalence of different cell types in their vaginal smears, n=12 \cite{Antoine}.
Principal component analysis grouped samples according to strains (Supplementary Fig \ref{supPCA} a) and analysis was performed separately on each of them \cite{Antoine}. 

Each of the three phases of the EC were once considered as the control and in it TTMap identified the number of outliers; among the 24 estrous samples (Fig \ref{figure4} a, arrowheads), the 23 diestrous or the 23 proestrous (Fig \ref{supPCA}). 

In analogy with the previous work \cite{Antoine}, we first analyzed Balb-C samples and C57/BL6 separately and made three comparisons (E vs D, E vs P and D vs P). A higher overlap between the significant genes in the two strains was found with TTMap compared to the standard analysis \cite{Antoine} with increases from 0 \% to 20\% for estrous versus diestrous (Fig \ref{figure4} b). An increase from 5 (Balb-C) and 19 \% (C57-BL6) to 28 and 32 \% in the comparison between estrous and proestrous and a similar result from 18 (Balb-C) and 47 \% (C57-BL6) to 36 and 45 \% in the comparison between diestrous and proestrous(Fig \ref{figure4} b). 

A high number of significant genes were the same between the common genes of the analyses done separately ("Separate", Fig \ref{figure4} c) on Balb-C and C57-BL6 and when the strains are pooled into one analysis and considered as batches ("Grouped", Fig \ref{figure4} c). 

The heatmaps of the deviation components of the missed genes once samples are grouped into one TTMap analysis, revealed that these genes were
differentially expressed significantly but into opposite directions in Balb-C versus C57-BL6 (on the left Fig \ref{figure4} c, blue negative regulation, yellow positive regulation). These genes have therefore an opposite role in the two strains through the EC. 
On the other hand, by looking at the heatmaps of the deviation components of the missed genes when samples are separated into two TTMap analyses and then overlapped, we discovered that all these genes are differentially expressed in the same direction for the two strains but not to the same extent, and hence did not reach significance either in Balb-C or C57-BL6  (on the right Fig \ref{figure4} c). This suggests that one can pool analyses of two batches into one analysis and gain important information, as genes that vary in the same direction, but not in the same extent and loose unwanted features, as genes that vary significantly but in opposite direction. The relevance of these missed genes is illustrated by the pathway analysis \cite{Panther} of the common genes revealing a "positive regulation of tumour necrosis factor (TNF) superfamily cytokine production" (Fig \ref{figure4} d, Supplementary Fig \ref{supPCA} c and d), missed by standard tools, which is relevant as TNF is also increased through the human menstrual cycle \cite{TNFalpha}.

Based on extent of deviation from the control group, TTMap orders subgroups within each phase (example for proestrous, where the subphases were labelled P1 to P5, P1 being the closest to the control and P5 the furthest, in Fig \ref{figure4} e, where estrous is the control). The significant genes in these subgroups are both known and previously unreported genes that vary throughout the EC of mice (Fig \ref{figure4} e, Supplementary Fig \ref{supPCA} b). 
For instance, TTMap confirmed the interferon signature found in \cite{Antoine} illustrated by the gene \textit{Irf7} but detects also missed genes such as \textit{Mybpc1}, a progesterone target gene \cite{Khan} also shown to be differentially expressed through the human menstrual cycle \cite{Pardo} or \textit{Lalba}, \textit{Csn3} two milk proteins (Fig \ref{figure4} e, Supplementary Fig \ref{supPCA} b). It is apparent that these missed genes have a significant deviation only in subgroups of the proestrous phase compared to estrous, as for instance in P1, P4, P5 for Lalba, with a deviation representing a log fold change bigger than -1, P5 for Csn3 with a deviation bigger than -2, P2, P4, and P5 for Mybpc1 with a fold change bigger than 1.2. In contrast, \textit{Irf7} which was not missed by standard tools has at least 1.2 fold change difference in all subgroups of proestrous. This provides an explanation of why they are missed by standard tools. By searching for the first, overall closest group to control (for instance P1), TTMap also spots samples that are in-between the two phases as illustrated by four estrous samples that are close to diestrous (Supplementary Fig \ref{supPCA} b, arrow).\\

\section{Discussion}
We have developed a {topology-based} clustering tool, Two-Tier Mapper (TTMap) that outperforms existing clustering tools especially when dealing with small sample numbers. TTMap calculates and relates individual deviation from a given control group.

The method includes an improved and extended version on the Mapper algorithm. By its unusual two-tier cover, we have rendered the algorithm theoretically stable with respect to various modifications of the data. The stability and accuracy of TTMap were validated on both \emph{in silico} data and real data. TTMap gives an individual profile of deviation compare to control and relates that to other samples which opens a new perspective for personalized medicine. It is able to face highly variable datasets as illustrated by the discovery of transcriptomic subgroups and outliers of the three phases of the estrous cycle relating possible alterations of hormone levels, rendering a refined description of it.

While previous Mapper applications require selection of multiple parameters that are problem dependent and can hence not be automated \cite{scRNAseqTOP}, \cite{Viral}, \cite{BC}, \cite{ReviewTop},\cite{SpinalCordInjury}, we have optimized TTMap's parameter selection and made it user-independent for global gene expression analysis.

A filter function provides the user with additional information about the composition in terms of quartiles of this function on the global clusters. As implemented here, the filter function takes into account only one specific aspect of refinement. To further enhance the method, one could filter by appropriate metadata such as categorical information or numerical data. All outputs can be compared as the global clusters are independent of the chosen filter function, providing a common reference for all outputs.  \\

TTMap is applicable to other types of data such as proteomic, metabolomic, or even neurological data, such as activity measurements in brain regions as the filter function, the mismatch distance, and the epsilon parameter can be changed and adapted by the user to cover specific needs (unpublished observations). \\

 \begin{Addendum}
\item[\textbf{Acknowledgement}] R.J. was funded by the SNF (SNF 31003A-162550), M.C. and S.O. are supported in part by ERC grant Gudhi (ERC-2013-ADG-339025)\\
  \item[\textbf{Contribution}] R.J. developed the software and did the analyses with the help and supervision of J.R., K.H. and C.B. and J.R. suggested the data to analyse and helped familiarise with it. R.J, M.C. and S.O. analysed the stability of the method and the theoretical aspects. All the authors contributed to the manuscript. R.J. came up with the original idea. C.B. supervised the biological part of the work and K.H. the mathematical.\\
 \item[\textbf{Competing Interests}] The authors declare that they have no
competing financial interests.
 \item[\textbf{Correspondence}] Correspondence 
should be addressed to R. Jeitziner~(email: rachel.jeitziner@epfl.ch).
  \end{Addendum}

\newpage
\newpage
\newpage
 \begin{figure*}[h]
\begin{center}
\includegraphics[height= 17.5cm]{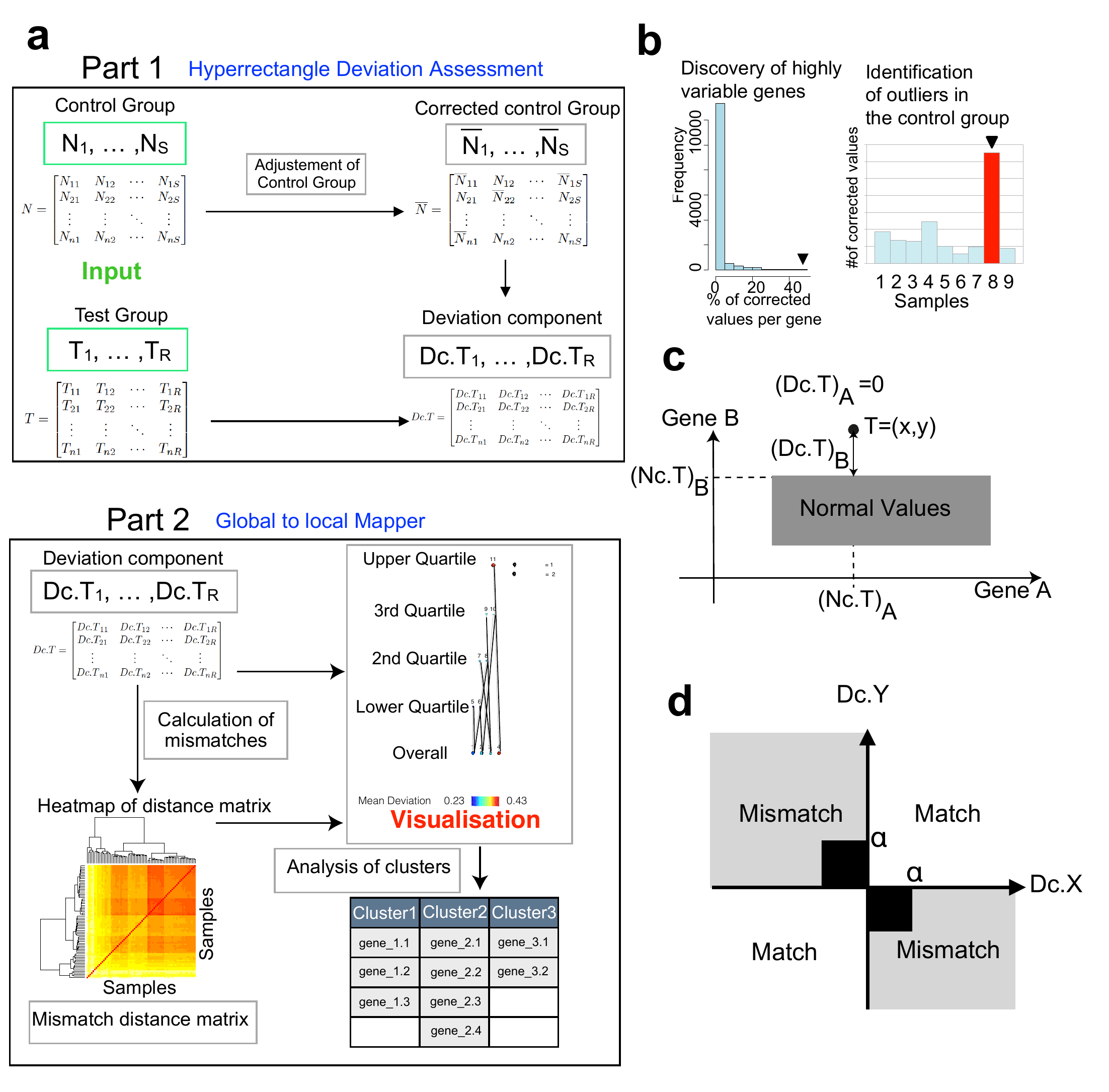}
\caption{Schematic overview of TTMap.}
\label{explanationfig}
\end{center}
\end{figure*}
\begin{figure}[t]
 \contcaption{ Schematic overview of TTMap. (a) The inputs (green) are given by two gene expression matrices, the control ($\mathsf{N}$) and the test group ($\mathsf{T}$), rows represent genes and columns samples. In Part 1, TTMap adjusts the control group for outlier values ($\bar{N}_*$), feature by feature. It calculates deviation from this corrected control group for individual samples in the test group ($Dc.T_*$). In Part 2, TTMap computes a similarity measure, the mismatch distance (represented as a heatmap) using the deviation components. The Mapper \cite{Mapper} algorithm is used with a two-tier cover to generate a visual representation of the clustering creating a network of global clusters (Overall) and local clusters (1st, 2nd, 3rd, 4th quartile of a filter function). It takes as inputs the mismatch distance and the deviation components. (b) Possible outputs after the first part of TTMap: histogram representing the frequency of features per percentage of outliers (left) and a barplot of the number of outliers per sample in the control group (right) to enable the discovery of highly variable genes or samples (red, arrow). (c) Scheme of a test sample $T$ together with its deviation components $Dc.T = (Dc.T_A,Dc.T_B)$ and normal component $Nc.T = (Nc.T_A,Nc.T_B)$ from the hyperrectangle (box) of normal values, example for $n=2$ genes A and B d) Scheme defining a match and a mismatch between two deviations components (Dc) of test samples X and Y with cutoff $\alpha$ to remove noise close to $0$ ($n =1$). The mismatch distance between two samples is the sum of mismatches through all the genes.}
\end{figure}

\newpage

 \begin{figure*}[h]
\begin{center}
\includegraphics[height= 22cm]{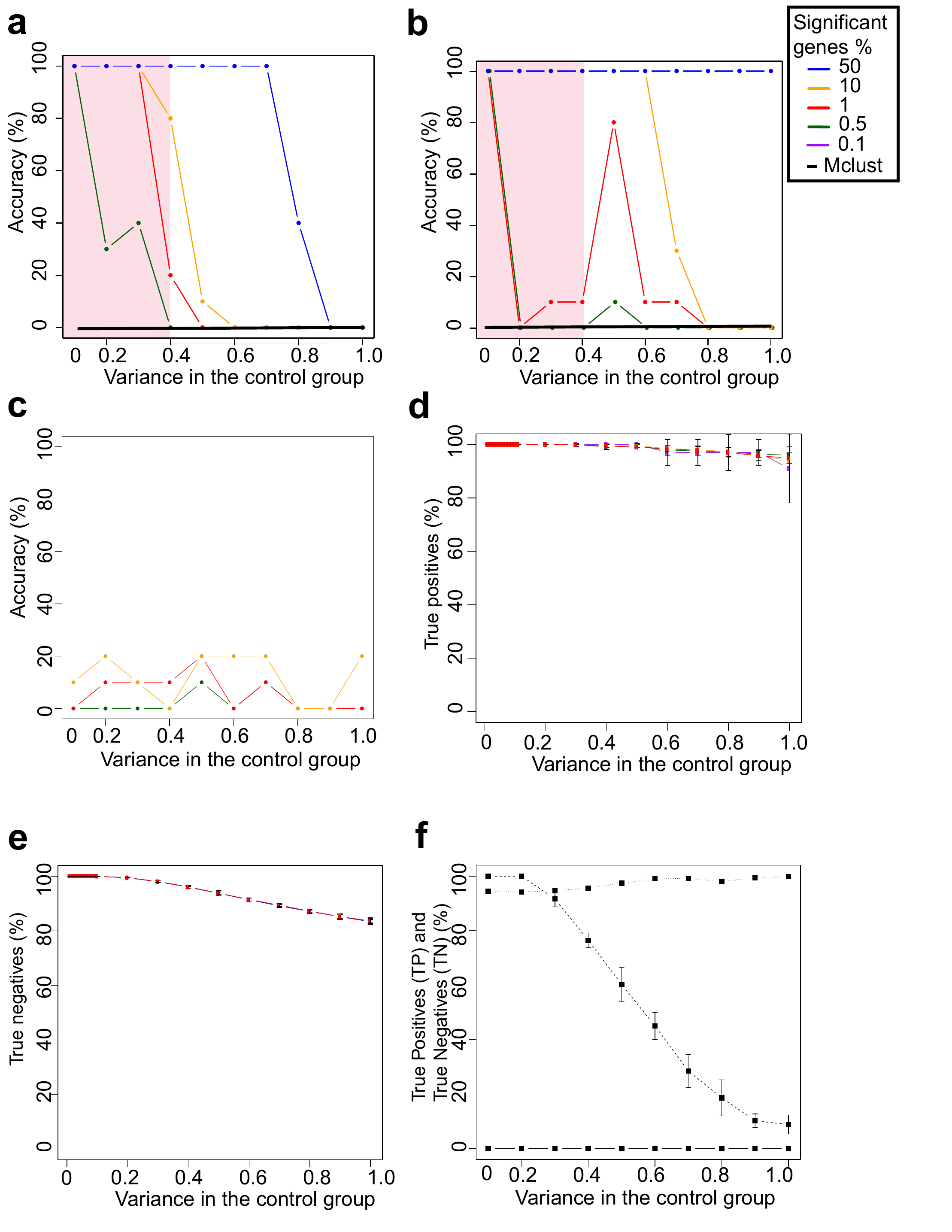}
\caption{In silico validation of TTMap.(continued on the next page)}
\label{figure2}
\end{center}
\end{figure*}
\begin{figure}[t]
 \contcaption{ In silico validation of TTMap. (a) Plot showing the accuracy of TTMap in percentage of time it correctly finds the right subgroups on an in silico dataset over a range of different variances, $N > 30$ individual curves were established for different percentages of significant genes, the accuracy of Mclust on the same dataset is shown in black, using epsilon with probability 0.025 (b) using epsilon with probability 0.975 (c) Plot showing Mclust on the deviation components $N=10$ per condition. (d) Percentage of true positives and (e) true negatives when the right groups are found $N> 30$ per condition (f) True positives (TP) and True negatives (TN) using moderated-t-test when the correct groups are given and when they are unknown. }
 \end{figure}
 \newpage
 \begin{figure*}[h]
\begin{center}
\includegraphics[height=17cm]{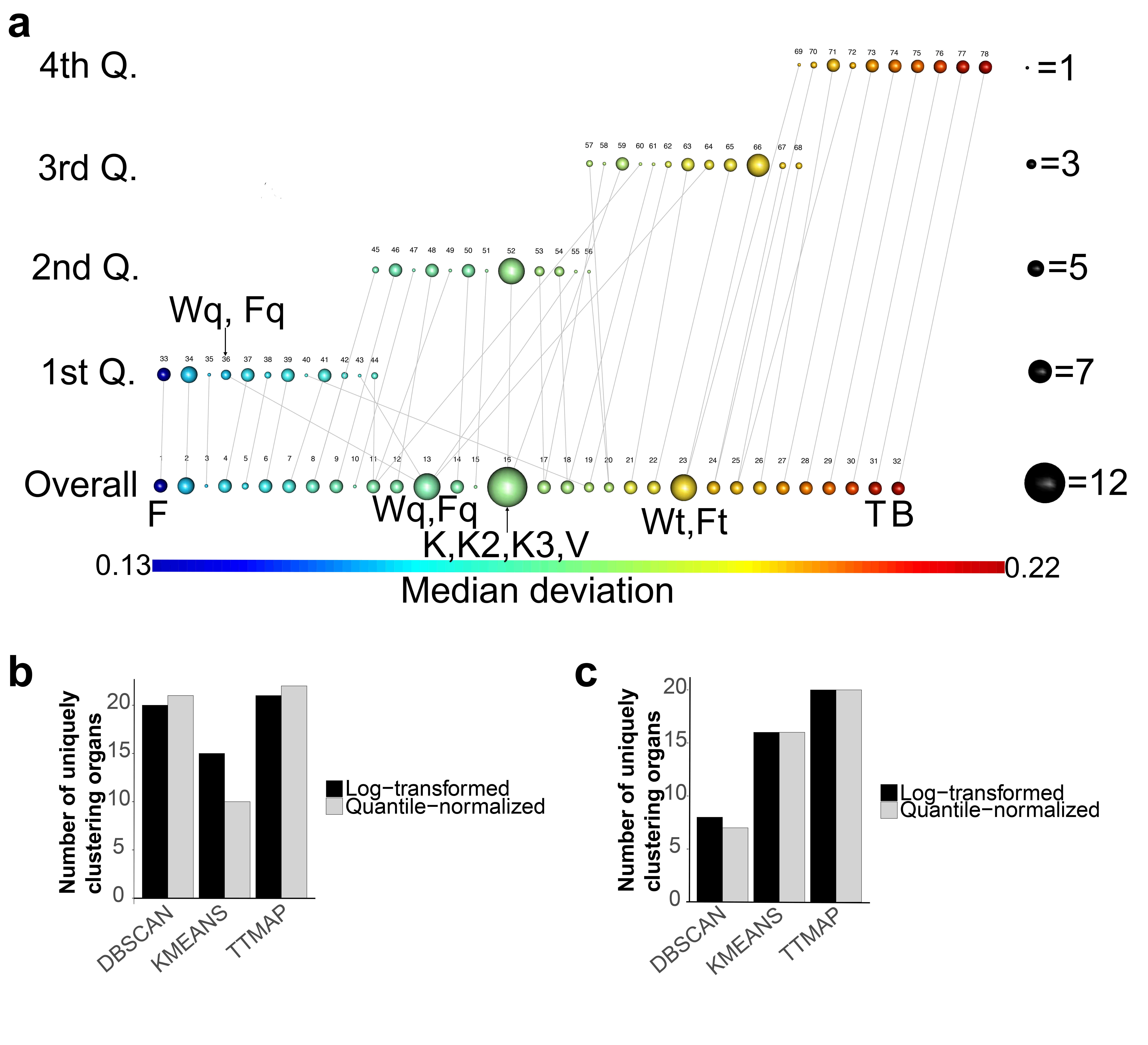}
\caption{Validation of TTMap on a well-characterized dataset from the fly-atlas (continued on the next page).}
\label{figure3}
\end{center}
\end{figure*}
\begin{figure}[t]
 \contcaption{ TTMap characterize deviations of organs from whole fly tissues. (flyatlas: GSE7763). (a) Output of TTMap showing the global clusters (Overall) that capture overall differences and local clusters within each quartile (1st, 2nd, 3rd, 4th Q., legend on the left) of the amount of deviation function with its links to the global clusters. The size of the sphere corresponds to the number of samples in the cluster (legend of the size of spheres on the right), the color the average amount of deviation. The number above the sphere is an identification and the letter under it is reflecting the organs that are in that cluster (K: spermatacea virgin, K2: spermatacea mated and K3: spermatacea virgin (redone), V: adult toracic muscle, Wq: fatbody of the wandering larvea, Fq: fatbody of the feeding larvea, F: whole larvea, T: Testes, B: Brain). On the bottom a legend of the color code, representing the mean amount of deviation.  (b) Barplot representing the number of uniquely clustering organs on log-transformed data and on quantile-normalised data for DBSCAN, Kmeans and TTMap to measure the stability of the methods by normalisation. (c) Barplot representing the number of uniquely clustering organs when the data is randomly subselected for  50\% of the genes on log-transformed data and on quantile-normalised data for DBSCAN, Kmeans and TTMap to measure the stability of the methods by random subselection.}
\end{figure}
 \begin{figure*}[h]
\begin{center}
\includegraphics[height= 22cm]{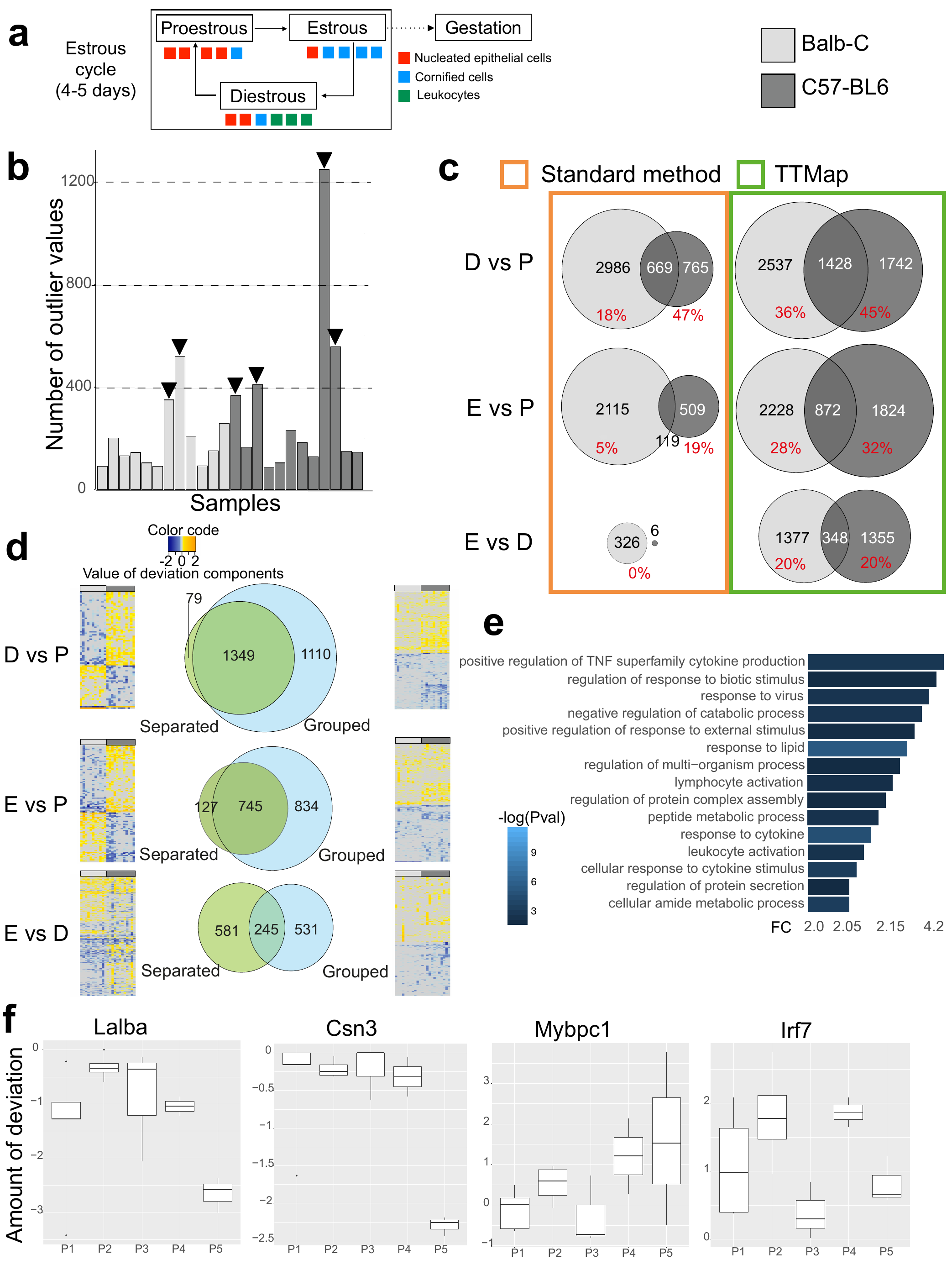}
\caption{Estrous cycle related gene expression changes in the mammary glands of C57Bl6 and BalbC mice. (continued on the next page).}
\label{figure4}
\end{center}
\end{figure*}
\begin{figure}[t]
\contcaption{Estrous cycle related gene expression changes in the mammary glands of C57Bl6 and BalbC mice. (a) Scheme of mice estrous cycle lasting around 4-5 days. The estrous cycle is divided into Proestrous (P) followed by Estrous (E) and then by Diestrous (D) phase, determined according to the prevalence of different cell types (nucleated epithelial cells, cornified cells, leukocytes) in vaginal cytology.  After E, mice can undergo gestation and D is skipped. (b) Barplot representing the number of outlier values in the control group (estrous phase) per sample. Samples having a high number of outlier values and a prevalence to remain isolated during clustering when E is the test group are outliers (arrowhead).  (b) Venn diagrams of the differentially expressed genes between E vs P, D vs E, and E vs D using standard analysis tools (orange) and TTMap (green) on Balb-C compared to C57-BL6 analysed separately. In red, the fraction of common significant genes per strain (\% over total number of significant genes). (c) Venn diagrams of the common differentially expressed genes when the analysis is done separately on the two mouse strains (Separated) or with the two mouse strains combined (Grouped) into one analysis using TTMap. Adjacent to them heatmaps of the deviation components illustrating in each situation the reason why the genes were missed.  (d) Panther pathway analysis \cite{Panther} of significant genes identified by TTMap in the comparison D vs P shown by Fold Change (FC) of importance of the pathway with \textit{-log(Pval)} as a color code (e) Boxplot representing the deviation component values in the identified subgroups of P (P1, P2, P3, P4, P5) by TTMap ordered by amount of deviation compared to the estrous samples (controls) of the genes \textit{Lalba}, \textit{Csn3}, \textit{Mybpc1} and \textit{Irf7}.}
\end{figure}

 \section{Online Methods}
 \label{Online}
 \subsection{Data preprocessing/input formats}

Prior to the analysis, the collected data are log-transformed and grouped into two separate tables, where columns are samples and rows are features from: 
\begin{itemize}
\item a group $\mathsf N$, called the \emph{normal (or control) group } the elements of which are denoted $N_{1}, \dots, N_{S}$, where $S$ is the number of collected samples in this group.  
\item a group $\mathsf T$, called the \emph{test group}, the elements of which are denoted $T_{1}, \dots, T_{R},$ where $R$ is the number of collected samples in this group.
\end{itemize}
The number of features measured (e.g., number of genes expression levels of which were determined) in each sample is written $n$. Thus, each element in group $\mathsf T$ and group $\mathsf N$ is a vector in $\mathbb{R}^n$.\\
If {different numbers of features have been measured for groups $\mathsf T$ and $\mathsf N$}, then {HDA considers only} the features measured across the whole data set.\\

\subsection{Data outputs}
The following files are being produced in the analysis :
\begin{enumerate}
\item[File 1.] standardized matrix $\bar{N}.$
\item[File 2.] number of modified features per sample, indicating outliers in the control group.
\item[File 3.] number of modified samples per feature, giving the features that change in the control group. 
\item[File 4.] deviation components  per test sample $Dc.T_1, \dots,Dc.T_R,$ indicating the pattern of deviation compared to the control.
\item[File 5.] normal components of test samples $Nc.T_1, \dots ,Nc.T_R.$
\item[File 6.] distance matrix.
\item[File 7.] visual representation of the clustering, giving subgroups in the test samples.
\item[File 8.] description of the clusters in File 7, with information.
\end{enumerate}

\subsection{Parameter selection}
The following parameters are computed and can be changed
\begin{itemize}
\item $e$ : can be changed by the user, the default value is $e$ is the $90^{\text{th}}$ percentile of the standard deviations for every feature multiplied by $\frac{2}{\sqrt{S}}$, where $S$ is the number of samples in the control group. 
\item If the user wants to remove features that have more than p\% of NAs in the control group. The  parameter $p$ is set to 100 \% by default.
\item $\epsilon$ : Assuming that the two vectors $X$ and $Y$ follow the same normal distribution $\mathcal{N}(\mu_i, \sigma_i^2)$ for feature $i$, the parameter $\epsilon$ is estimated using the data. Feature by feature the probability to be a mismatch is calculated. Let $X_k$ be the random vector representing the gene expressions of a sample $T_k$. Therefore, let
 $$p_{1j,k} = P(X_{kj} < \min_{i \in \mathcal{I}(T_{k})}({\overline{N}}_i)_j),\textrm{ the probability to be underexpressed compared to normal values}$$
 $$p_{2j,k}= P(X_{kj} >\max_{i \in \mathcal{I}(T_{k})}({\overline{N}}_i)_j),\textrm{ the probability to be overexpressed compared to normal values}$$
 $$p_{3j,k} = P( \min_{i \in \mathcal{I}(T_{k})}({\overline{N}}_i)_j <X_{kj}  < \max_{i \in \mathcal{I}(T_{k})}({\overline{N}}_i)_j), \textrm{ the probability to be inside the normal range} $$
 Then, we define
 $$p_{1j,k}^{\alpha} = P(X_k < \min_{i \in \mathcal{I}(T_{k})}({\overline{N}}_i)_j-\alpha)),$$
 $$p_{2j,k}^{\alpha}  = P(X_k>\max_{i \in \mathcal{I}(T_{k})}({\overline{N}}_i)_j+\alpha)).$$
 Hence the probability $(P_k^l)_j$ of a mismatch between the $j$-th gene of $(X_k,X_l)$ is equal to : 
 $ ((p_{3j,k}+p_{1j,k} )\cdot p_{2j,l}^\alpha) +(p_{3j,k}+p_{2j,k})\cdot p_{1j,l}^\alpha)+((p_{3j,l}$ \\$+p_{1j,l})\cdot p_{2j,k}^\alpha) +(p_{3j,l}+p_{2j,l})\cdot p_{1j,k}^\alpha)-p_{1j,k}^\alpha\cdot p_{2j,l}^\alpha -p_{2j,k}^\alpha\cdot p_{1j,l}^\alpha,$ where for example $((p_{3j,k}+p_{1j,k} )\cdot p_{2j,l}^\alpha)$ would represent the probability that $X_k$ for gene $j$ is either as the control $(p_{3j,k})$ or lower than the control $(p_{1j,k})$ whereas $X_l$ is marginally (more than alpha) higher than the control $(p_{2j,l}^\alpha)$, and so it represents a mismatch.
 Using Chern-Stein's theorem, it is known that if $n>>S$, and if the probabilities accumulated around 0 as is the case for gene expression data, then the sum over all genes of mismatches follows a Poisson distribution with mean $\sum_{j=1}^n (P_k^l )_j$.
 This in turn allows one to determine how significant the number of mismatches between X and Y is, if both vectors follow the same distribution. Hence, $\epsilon$ is given by $P(\sum_{j=1}^n (P_k^l )_j < \epsilon) = \beta,$ which can be obtained from the quantiles of a Poisson law. Thus, samples are linked if the number of mismatches between them is less than $\epsilon$, which is the $\beta \%$ confidence threshold of mismatches for samples following the same distribution.\\
 If $\epsilon$ is chosen such that $P(\sum_{j=1}^n (P_k^l )_j < \epsilon) = 0.025$, it means that only in $2.5\%$ of the cases if $X_k$ and $X_l$ are distributed in the same way, they would have such a small number of mismatches and therefore it is certain that $X_k$ and $X_l$ must be clustered together. In the same way, if  $\epsilon$ is chosen such that $P(\sum_{j=1}^n (P_k^l )_j < \epsilon) = 0.975$, it means that only in $97.5\%$ of the cases if $X_k$ and $X_l$ are distributed in the same way, they would have such a high number of mismatches and therefore it is certain that $X_k$ and $X_l$ must be separated.
The user can therefore choose either to cluster samples together only if one is sure that samples should be clustered together (0.025) or choose to separate samples only if one is sure that samples need to be separated or finally the user has the option to put another value for parameter $\epsilon$, when the \% of mismatches to be expected is already known. \\
\item Distance: Alternative distances, such as correlation distance, Euclidean distance, useful when there is no control group, and complete mismatch distance, a stringent version of the mismatch distance defined above, are implemented in GLMap and can be selected. Of note, in those cases the parameter $\epsilon$ needs to be adapted and has no appropriate default value. The mismatch distance is appropriate for gene expression data, since it captures deviation of samples from the control values with the same orientation, regardless of the magnitude of deviation.\\
 \item The default filter function can be changed by the user and any metadata can be taken as input, which needs to be a vector of the same length as the number of samples. 
\item $S$: If the user is interested in the deviation of a specific subset of {features, e.g.,} genes linked to a certain pathway, then the set $S$ can be modified appropriately, it is provided as a vector of gene identifications.\\
\end{itemize}

\subsection{Data sources}

Drosophila Affymetrix array data files were downloaded from GEO accession no GSE7763. Mouse data was kindly provided by A. Snijders and colleagues \cite{Antoine}.

\subsection{Synthesised data}
Since microarray gene expression data, is modelled as a normal distribution (\cite{Noiseestimation}), the simulated data has been generated as follows. For a fixed natural number $m$ less than 10,000,  {K random lists of 10,000 real numbers} are generated each, where
   $C_1, \dots C_{K/2}$ are the $K/2$ controls and $TA_1 , TA_2, \dots, TA_{K/4}, TB_1, TB_2, \dots, TB_{K/4},$ are the test samples, each with 10,000 genes. The subgroups $TA$ and $TB$ have a mean per gene that is $\Delta$ times higher, respectively lower than the mean of the control in $m$ genes.   Hence,     $$(C_1)_i, \dots, (C_{K/2})_i \in \mathcal{N} (\mu, \sigma^2)$$ for all $1\leq i\leq 10,000$, and $$(TA_1)_i, (TA_2)_i, \dots (TA_{K/4})_i , (TB_1)_i, (TB_2)_i, \dots, (TB_{K/4})_i\in \mathcal{N} (\mu, \sigma^2)$$ for all $1\leq i\leq  10,000-m$, while $$(TA_1)_i, (TA_2)_i, \dots, (TA_{K/4})_i \in \mathcal{N} (\mu+ \Delta, \sigma^2),$$ 
and $$(TB_1)_i, (TB_2)_i, \dots, (TB_{K/4})_i \in \mathcal{N} (\mu - \Delta, \sigma^2)$$ for all  $10,000-m<i \leq10,000$ and $K$ is either $12$, $200$ or $400$.\\
 
\subsection{Code availability}
TTMap is implemented as an open-source R package under revision at the Bioconductor.

\subsection{Theoretical part}
In this section, all functions are assumed to be of Morse type, as defined in~\cite{Carlsson09b}.
This mild assumption is purely technical and assures that the mathematical objects we deal with are well defined. All the assumptions made in the stability theorems are verified concerning TTMap subsequently in section \ref{verification}.

\subsubsection{Mathematical background}

\paragraph{Reeb Graphs}
\label{sec:ReebGraphDef}
Given a topological space $X$ and a continuous function 
$f:X\rightarrow\R$, we define the equivalence relation $\sim_f$ between points of $X$ by:
%
\begin{align*}
x\sim_f y\ &\Longleftrightarrow\ f(x)=f(y) \emph{ and } x,y\ \emph{belong to the same} \\
& \emph{connected component of}\ f^{-1}(f(x))=f^{-1}(f(y)).
\end{align*}
The \emph{Reeb graph}~\cite{Reeb46}, denoted by $\Reeb_f(X)$, is the quotient space $X/\sim_f$.


As $f$ is constant on equivalence classes, there is an induced map $\tilde{f}~:~\Reeb_f(X)~\rightarrow~\R$ such that $f~=~\tilde{f}~\circ~\pi$,
where $\pi$ is the quotient map $X\to \Reeb_f(X)$.
%
%
If $f$ is a function of Morse type, then
the Reeb graph is a multigraph~\cite{deSilva16}, whose nodes are in one-to-one correspondence with the connected components
of the critical level sets of $f$. 


\paragraph{Extended Persistence}
\label{sec:Persistence}

Given any Reeb graph $\Reeb_f(X)$,
the so-called {\em extended persistence diagram} $\Dg(\tilde f)$ is a multiset of points in the
Euclidean plane $\R^2$ that can be computed with {\em extended persistence theory}~\cite{Cohen09, Chazal16}.
Each of its points has a specific {\em type}, which is either $\Ord_0$, $\Rel_1$, $\Ext_0^+$ or $\Ext_1^-$.
Orienting the Reeb graph vertically so $\tilde{f}$ is the height function,
we can see each connected component of the graph as a trunk with
multiple branches, some oriented upwards, others oriented downwards
and holes.  The following correspondences are obtained, where the
{\em vertical span} of a feature is the span of its image by~$\tilde{f}$:
\begin{itemize}
\item The vertical spans of the trunks are given by the points in $\Ext_0^+(\tilde{f})$;
\item The vertical spans of the branches that are oriented downwards are given by the points in $\Ord_0(\tilde{f})$;
\item The vertical spans of the branches that are oriented upwards are given by the points in $\Rel_1(\tilde{f})$;
\item The vertical spans of the holes are given by the points in $\Ext_1^-(\tilde{f})$. 
\end{itemize}
%
These correspondences
provide a dictionary to read off the structure of the Reeb graph from
the corresponding extended persistence diagram (Figure~\ref{fig:sign}.a).
Note that it is a bag-of-features type descriptor, taking an inventory of all
the features (trunks, branches, holes) together with their vertical
spans, but leaving aside the actual layout of the features. As a
consequence, it is an incomplete descriptor: two Reeb graphs with the
same persistence diagram may not be isomorphic as combinatorial graphs or as metric graphs.



\subsubsection{Generalized structure of TTMap}

Let $X$ be a topological space and let $f:X\to\R$ be a Morse-type
function. Consider a family of pairwise disjoint intervals of~$\R$ with
non-empty interiors, such that
the union of all the intervals is still an interval.  Add $\R$ to this
family and call the result~$\I$.
Considering the class of Morse-type pairs $(X,f)$ such that $\I$ is a cover
of~$\im(f)$, our aim is to study the structure of $\Mapper(X,f,\I)$
and its stability with respect to perturbations of~$(X,f)$ within this class.
Note that, $TTMap(Dc.T,\tau, \mathcal{I})$ is a special case of $\rm M(P,f,\mathcal{I})$, where $P$ is given by ${\rm Dc.T}$ and $X$ is the corresponding underlying support, $f$ is given by $\tau$ and $\I$ is given by the quantiles $q_{ab}$ and the real line and $\delta$ is given by the parameter $\epsilon$ (\ref{GLMap}).

\begin{defin}
We define the following descriptor for $\Mapper(X,f,\I)$: 
$$\DM(X,f,\I):=(\Dg(\tilde{f}),\phi, \{\Delta_I\}_{I\in\I}),$$ 
where: \begin{itemize}
\item $\phi:\Dg(\tilde{f})\rightarrow \Ext_0^+(\tilde{f})$ maps a persistence pair (i.e. $(a,b)$ where $a$ and $b$ are the birth and the death time respectively of a topological feature to
the connected component of $X$ to which its corresponding feature belongs,
\item $\Delta_I=\{(x,x)\ |\ x\in I\}$ is the diagonal subset of  $I\times I$.
\end{itemize}
\end{defin}
\noindent Intuitively,  $\Mapper(X,f,\I)$ can be reconstructed from $\DM(X,f,\I)$
in 3 steps (Figure~\ref{fig:sign}.c, d, e and f):
\begin{enumerate}
\item  Create one super-node per point in $\Ext_0^+(\tilde{f})$. 
\item For each interval~$I\in\I$, create one node per point $(x,y)\in \Dg(\tilde{f})$ such that~$I$ is contained entirely in the lifespan of~$(x,y)$, 
which is materialized in the descriptor~$\DM(X,f,\I)$ by the fact that the line segment $\Delta_{(x,y)}$ bounded by the horizontal and vertical 
projections of~$(x,y)$ onto the diagonal $\Delta$ contains~$\Delta_I$. 
If $(x,y)\in \Ord_0(\tilde f)\cup\Rel_1(\tilde f)\cup\Ext_0^+(\tilde f)$ then create a vertex also if $I$ contains~$x$. If $(x,y)\in \Ext_0^+(\tilde f)$ then create a vertex also if $I$ contains~$y$. 
\item Draw the links prescribed by $\phi$ between the super-nodes and the rest of the nodes.
\end{enumerate}
%
%
\begin{thm}\label{completeness}\label{th:complete}{\bf Completeness.}
$\DM(X,f,\I)$ is a complete descriptor of $\Mapper(X,f,\I)$.
\end{thm}
%
%

%
\begin{proof}

At any level~$\alpha\in\R$, the following equality holds:
\begin{equation}\label{eq:level}
\begin{split}
&\#\left\{C\ :\ C{\rm\ is\ a\ connected\ component\ of\ }\tilde{f}^{-1}(\{\alpha\})\right\}
\ =\   \\
 & \#\left\{(x,y)\in\Dg(\tilde{f})\ :\ \alpha\in\textrm{lifespan}\,(x,y)\right\},
 \end{split}
\end{equation}
where: 
\[
\textrm{lifespan}\,(x,y)\ =\ \left\{\begin{array}{ll}
[x,y] & \mbox{if $(x,y)\in\Ext_0^+(\tilde f)$}\\[0.5ex]
(y,x) & \mbox{if $(x,y)\in\Ext_1^-(\tilde f)$}\\[0.5ex]
[x,y) & \mbox{if $(x,y)\in\Ord_0(\tilde f)$}\\[0.5ex]
(y,x] & \mbox{if $(x,y)\in\Rel_1(\tilde f)$}
\end{array}\right.
\]

Indeed, let $\alpha \in\R$.  Assume for simplicity that
$\alpha\not\in\Crit(f)$ (if $\alpha\in\Crit(f)$ then the same analysis
holds with the extra technicality that the type of each interval endpoint, open or closed, must be taken into account).
Define the following quadrants (Figure~\ref{fig:sign}.b):
\begin{align}
Q^\alpha_{\rm NW}&=\{(x,y)\in\R^2\ :\ x\leq\alpha{\rm\ and\ }y\geq\alpha\}  \nonumber \\
Q^\alpha_{\rm NE}&=\{(x,y)\in\R^2\ :\ x\geq\alpha{\rm\ and\ }y\geq\alpha\}\nonumber\\
Q^\alpha_{\rm SW}&=\{(x,y)\in\R^2\ :\ x\leq\alpha{\rm\ and\ }y\leq\alpha\} \nonumber \\
 Q^\alpha_{\rm SE}&=\{(x,y)\in\R^2\ :\ x\geq\alpha{\rm\ and\ }y\leq\alpha\}\nonumber
\end{align}
%
%
%
Since points in $\Ord_0(\tilde f)$ and $\Ext_0^+(\tilde f)$ are located above the diagonal and 
points in $\Ext_1^-(\tilde f)$ and $\Rel_1(\tilde f)$ are located below,
proving Equation~(\ref{eq:level}) amounts to showing that 
\begin{equation}\label{eq:tbp}
\begin{split}
& {\rm dim}\left(H_0\left(\tilde f^{-1}(\{\alpha\})\right)\right)=|\Ord_0(\tilde f)\cap Q^\alpha_{\rm
  NW}|  \\
  & + |\Ext_0^+(\tilde f)\cap Q^\alpha_{\rm NW}| + |\Ext_1^-(\tilde
f)\cap Q^\alpha_{\rm SE}| + |\Rel_1(\tilde f)\cap Q^\alpha_{\rm
 SE}|.
  \end{split}
\end{equation}
 For this the Mayer-Vietoris theorem is used with spaces $A=\tilde
 f^{-1}((-\infty,\alpha])$, $B=\tilde f^{-1}([\alpha,+\infty))$,
    $A\cap B = \tilde f^{-1}(\{\alpha\})$, and $A\cup B =
    \Reeb_f(X)$. This theorem can be used because the
    Morse-type condition implies that $A, B$ are deformation retracts
    of neighborhoods $A', B'$ in $\Reeb_f(X)$ with $A'\cap B'$
    deformation retracting onto $A\cap B$. Hence, the following sequence is 
    exact:
\begin{align}
H_2(\Reeb_f(X)) & \overset{\partial_2}{\longrightarrow} H_1\left(\tilde f^{-1}(\{\alpha\})\right) \nonumber \\
& \overset{\phi}{\longrightarrow} \overbrace{H_1\left(\tilde f^{-1}((-\infty,\alpha])\right)\oplus H_1\left(\tilde f^{-1}([\alpha,+\infty))\right)}^{K_1} \nonumber \\
& \overset{\psi}{\longrightarrow} H_1(\Reeb_f(X)) \nonumber \\
& \overset{\partial_1}{\longrightarrow} H_0\left(\tilde f^{-1}(\{\alpha\})\right) \nonumber \\
& \overset{\zeta}{\longrightarrow} \underbrace{H_0\left(\tilde f^{-1}((-\infty,\alpha])\right)\oplus H_0\left(\tilde f^{-1}([\alpha,+\infty))\right)}_{K_0} \nonumber \\
& \overset{\xi}{\longrightarrow} H_0(\Reeb_f(X)) \overset{\partial_0}{\longrightarrow} 0 \nonumber
\end{align}
To be more specific, exactness gives the following relations:
\begin{align*}
\im(\partial_2)&=\ker(\phi)\Label{eq:rhophi} & \im(\partial_1)&=\ker(\zeta)\Label{eq:zetachi} \\
\im(\phi)&=\ker(\psi)\Label{eq:phipsi} & \im(\zeta)&=\ker(\xi)\Label{eq:zetaxi} \\
\im(\psi)&=\ker(\partial_1)\Label{eq:psichi} & \im(\xi)&=\ker(\partial_0)\Label{eq:xisigma}
\end{align*}

It follows from (\ref{eq:xisigma}) and from \cite{Bauer14} that
\begin{equation}\label{eq:dimimxi}
\begin{split}
\dim(\im(\xi)) &=\dim(\ker(\partial_0))=\dim(H_0(\Reeb_f(X)))\\
& =|\Ext_0^+(\tilde f)|.
\end{split}
\end{equation}

Moreover, according to Theorem~2.9 in~\cite{Carriere15c}, we have $H_p(\Reeb_f(X)) = 0$ for any $p\geq 2$.  Using (\ref{eq:rhophi}), it follows that $\im(\partial_2) = 0 = \ker(\phi)$, hence 

\begin{multline}\label{eq:dimimphi}
\begin{split}
0&=\dim\left(H_1\left(\tilde f^{-1}(\{\alpha\})\right)\right)=\dim(\ker(\phi)) + \dim(\im(\phi))\\
&=\dim(\im(\phi)).
\end{split}
\end{multline}

Using equations (\ref{eq:zetachi}) to (\ref{eq:dimimphi}) and Theorem 2.5 in \cite{Carriere15c}, the following equalities hold:
\begin{align}
&\dim\left(H_0\left(\tilde f^{-1}(\{\alpha\})\right)\right)=\dim(\ker(\zeta)) + \dim(\im(\zeta)) \nonumber \\
& = \dim(\im(\partial_1)) + \dim(\ker(\xi)) \nonumber \\
& = \dim(H_1(\Reeb_f(X))) - \dim(\ker(\partial_1)) + \dim(\ker(\xi)) \nonumber \\
& = |\Ext_1^-(\tilde f)| - \dim(\im(\psi)) + \dim(\ker(\xi)) \nonumber \\
& = |\Ext_1^-(\tilde f)| - \dim(K_1) + \dim(\ker(\psi)) + \dim(\ker(\xi)) \nonumber \\
& = |\Ext_1^-(\tilde f)| - \dim(K_1) + \dim(\im(\phi)) + \dim(\ker(\xi)) \nonumber \\
& = |\Ext_1^-(\tilde f)| - \dim(K_1) + \dim(\ker(\xi)) \nonumber \\
& = |\Ext_1^-(\tilde f)| - \dim(K_1) + \dim(K_0) - \dim(\im(\xi)) \nonumber \\
& = |\Ext_1^-(\tilde f)| - \dim(K_1) + \dim(K_0) - |\Ext_0^+(\tilde f)| \nonumber
\end{align}

It remains to compute $\dim(K_1)$ and $\dim(K_0)$. Using the correspondence between connected components and branches
of $\Reeb_f(X)$ and points of $\Dg(\tilde f)$~\cite{Bauer14}, it holds that
\begin{align}
\dim(K_1)&=\dim\left(H_1\left(\tilde f^{-1}((-\infty,\alpha])\right)\right)\nonumber\\
& \ \ \ \ \  \ \ \ \ \ +\dim\left(H_1\left(\tilde f^{-1}([\alpha,+\infty))\right)\right)\nonumber\\
&=|\Ext_1^-(\tilde f)\cap Q^\alpha_{\rm SW}| + |\Ext_1^-(\tilde f)\cap Q^\alpha_{\rm NE}|
\end{align}
and 
\begin{align}
\dim(K_0)&=\dim\left(H_0\left(\tilde f^{-1}((-\infty,\alpha])\right)\right) \nonumber\\
&  \ \ \ \ \  \ \ \ \ \  +\dim\left(H_0\left(\tilde f^{-1}([\alpha,+\infty))\right)\right)\nonumber\\
&=|\Ord_0(\tilde f)\cap Q^\alpha_{\rm NW}| + |\Ext_0^+(\tilde f)\cap (Q^\alpha_{\rm NW}\cup Q^\alpha_{\rm SW})|\nonumber\\
& +|\Rel_1(\tilde f)\cap Q^\alpha_{\rm SE}| + |\Ext_0^+(\tilde f)\cap (Q^\alpha_{\rm NW}\cup Q^\alpha_{\rm NE})|. 
\end{align}
Combining these results, we obtain
\begin{align}
&\dim\left(H_0\left(\tilde f^{-1}(\{\alpha\})\right)\right) =|\Ext_1^-(\tilde f)| - |\Ext_1^-(\tilde f)\cap Q^\alpha_{\rm SW}| \nonumber \\
&- |\Ext_1^-(\tilde f)\cap Q^\alpha_{\rm NE}| + |\Ord_0(\tilde f)\cap Q^\alpha_{\rm NW}| + |\Rel_1(\tilde f)\cap Q^\alpha_{\rm SE}|  \nonumber \\
&+ |\Ext_0^+(\tilde f)\cap (Q^\alpha_{\rm NW}\cup Q^\alpha_{\rm SW})| + |\Ext_0^+(\tilde f)\cap (Q^\alpha_{\rm NW}\cup Q^\alpha_{\rm NE})|  \nonumber \\
&- |\Ext_0^+(\tilde f)|   \nonumber\\
& \ \ \ \ \ \ \ \ \ \  \ \ \ \ \ \ \ \ \ \ \ \ \ \ \ \ \ =  |\Ext_1^-(\tilde f)\cap Q^\alpha_{\rm SE}|\nonumber\\
& + |\Ord_0(\tilde f)\cap Q^\alpha_{\rm NW}| + |\Rel_1(\tilde f)\cap Q^\alpha_{\rm SE}| + |\Ext_0^+(\tilde f)\cap Q^\alpha_{\rm NW}|, \nonumber
\end{align}
which gives~(\ref{eq:tbp}) and thus proves Equation~(\ref{eq:level}). 

The theorem is proved using the three steps of the reconstruction scheme detailed before the statement \ref{th:complete}.

According to the one-to-one correspondence between the connected components of $\Reeb_f(X)$ and the points of $\Ext_0^+(\tilde f)$,
Step~1 ensures that there are as many super-nodes as there are connected components in $\Reeb_f(X)$.

Equation~(\ref{eq:level}) can be extended to intervals at no cost to prove that the number of vertices created in Step~2 and the number of
nodes in $\Mapper(X,f,\I)$ (apart from the super-nodes) is the same.

Finally, each node $v$ of $\Mapper(X,f,\I)$ corresponds to some connected component
of the preimage $f^{-1}(I)$ of some interval $I\in\I$. That connected component lies
entirely in some connected component~$X_i$ of $X$, therefore $v$ gets connected to the
super-node corresponding to~$X_i$ in $\Mapper(X,f,\I)$.
This is the only type of connection that matters for $\Mapper(X,f,\I)$, 
since every pair of intervals other than $\R$ in $\I$ has an empty intersection.
Since the connected component corresponding to $v$ belongs to at least one feature of $\Reeb_f(X)$, 
or equivalently one persistence pair of $\Dg(\tilde f)$, 
this proves that the links prescribed by $\phi$ in Step~3 and the ones of $\Mapper(X,f,\I)$ are the same.
\end{proof}
 
This result states that whenever two descriptors are the same, their corresponding TTMap graphs must also be the same and therefore for what follows the results are shown in terms of diagrams.

\subsubsection{Stability theorems}

Note that $\{\Delta_I\}_{I\in\I}$ induces the grid
$\left(\End(\I\setminus\R)\times\R\right) \cup
\left(\R\times\End(\I\setminus\R)\right)$, (Figure~\ref{fig:sign} e). Intuitively, the distances of the
points of $\Dg(\tilde{f})$ to this grid give the amount of
perturbation allowed to preserve the structure of $\Mapper(X,f,\I)$. 
Reciprocally, for a given amount of perturbation~$\e$, drawing a square of radius~$\e$ around each diagram point allows us 
to see which diagram points may change grid cells and how the structure of~$\Mapper(X,f,\I)$ is impacted. 

\begin{defin}
Let $f,g$ be two Morse-type functions defined on topological spaces $X,Y$.
The \emph{descriptor distance} between $\DM(X,f,\I)$ and $\DM(Y,g,\I)$ is:
\[
d(\DM(X,f,\I),\DM(Y,g,\I))=\inf_{\Gamma}\ \emph{cost}(\Gamma),
\]
where $\Gamma$ ranges over all partial matchings between $\Dg(\tilde{f})$ and $\Dg(\tilde{g})$ such that
$(p,p') \in \Gamma\Rightarrow(\phi(p),\phi(p'))\in\Gamma$. 
\end{defin}

%

	
\begin{thm}\label{thm:perturb_bis}{\bf Stability with respect to changes of the filter function. }
For any Morse-type functions $f,g:X\to\R$:
\[ 
d(\DM(X,f,\I),\;\DM(X,g,\I))\leq \|f-g\|_\infty. 
\]
\end{thm}
\begin{proof}

Decompose $X$ into its various connected components: $X=X_1\sqcup X_2 \sqcup ... \sqcup X_n$,
and let $f_i:=f|_{X_i}:X_i\rightarrow\R$ and $g_i:=g|_{X_i}:X_i\rightarrow\R$.
Note that $\Dg(f)=\Dg(f_i)\sqcup...\sqcup\Dg(f_n)$, and similarly for $g$
and the induced maps $\tilde f$ and $\tilde g$.
Thus, one can build a matching $\Gamma$ that preserves connected components by taking any matching for each pair of subdiagrams $\Dg(f_i),\Dg(g_i)$.
For instance, let us take for each pair $\Dg(f_i),\Dg(g_i)$ the matching achieving $d(\DM(X_i,f_i,\I),\;\DM(X_i,g_i,\I))$.
Call it $\Gamma_i$, and let $\Gamma=\bigcup_i \Gamma_i$.
Hence, the following inequalities hold: 

\begin{align}
d & (\DM(X,f,\I),\;\DM(X,g,\I)) \leq {\rm cost}(\Gamma) \nonumber \\
& \leq \underset{i\in\{1,...,n\}}{\max} {\rm cost}(\Gamma_i) \nonumber\\
& = \underset{i\in\{1,...,n\}}{\max} d(\DM(X_i,f_i,\I),\;\DM(X_i,g_i,\I)) \nonumber\\
&= \underset{i\in\{1,...,n\}}{\max} d^\infty_{\rm b}(\Dg(\tilde f_i),\Dg(\tilde g_i)){\rm\ since\ }X_i{\rm\ is\ connected}  \nonumber\\
&\leq \underset{i\in\{1,...,n\}}{\max} \|\tilde f_i-\tilde g_i\|_\infty \textrm{\ by\ the\ stability\ theorem~\cite{Cohen07}} \nonumber\\
&= \|\tilde f-\tilde g\|_\infty \nonumber \\
&= \|f-g\|_\infty \textrm{\ since\ the\ quotient\ maps\ }\tilde f \textrm{\ and\ }\tilde g \nonumber\\
& \textrm{\ \ \ \ \ \ \ \ \ \ \ \ \ \ \ \ preserve\ function\ values}. \nonumber
\end{align}
\end{proof}

\begin{thm}\label{thm:DStab_bis}{\bf Stability with respect to perturbations of the domain. }
Let $X$ and $Y$ be two compact Riemannian manifolds
or length spaces with curvature bounded above.  Denote by $\rho(X)$
and $\rho(Y)$ their respective convexity radii.  Let
$f:X\rightarrow\R$ and $g:Y\to\R$ be Lipschitz-continuous Morse-type
functions, with Lipschitz constants $c_f$ and $c_g$ respectively.
Assume $\distgh(X,Y)\leq\frac{1}{20}\,\min\left\{\rho(X),\rho(Y)\right\}$. 
Then, for any correspondence $C\in\mathcal{C}(X,Y)$ such that
$\epsilon_{\mathfrak{m}}(C)<\frac{1}{10}\,\min(\rho(X),\rho(Y))$,
\begin{align*}
d & (\DM(X,f,\I),\; \DM(Y,g,\I)) \nonumber \\ 
& \leq (9(c_f+c_g)+\min\{c_f,c_g\})\epsilon_{\mathfrak{m}}(C)+\epsilon_{\mathfrak{f}}(C), \nonumber 
\end{align*}
where $\epsilon_{\mathfrak{m}}(C)$ and $\epsilon_{\mathfrak{f}}(C)$  are the distance distortion and the functional distortion \cite{Burago01}.
\end{thm}
\begin{proof}
If there is a one-to-one matching between the connected components of $X$ and $Y$ induced by the correspondence achieving $\distgh(X,Y)$, then
the proof follows the same line as the proof of Theorem~\ref{thm:perturb_bis}.
The only difference in the proof is the use of Theorem~3.4 in~\cite{Carriere15b} instead of the stability theorem~\cite{Cohen07}.
If such a one-to-one matching does not exist, $\distgh(X,Y)$ is infinite and so is $\epsilon_{\mathfrak{m}}(C)$, hence \begin{align*}
d & (\DM(X,f,\I),\; \DM(Y,g,\I)) \nonumber \\ 
& \leq (9(c_f+c_g)+\min\{c_f,c_g\})\epsilon_{\mathfrak{m}}(C)+\epsilon_{\mathfrak{f}}(C), \nonumber 
\end{align*} still holds. 

\end{proof}

\begin{thm}\label{th:sig-approx}{\bf Stability with respect to point cloud approximations. }
Let $X$ be a submanifold of~$\R^d$ with positive
reach $r(X)$ and convexity radius $\rho(X)$.  Let $f:X\rightarrow\R$
be a Lipschitz-continuous Morse-type function, with Lipschitz constant
$c$. 
Let $P\subseteq X$ be such that every point of~$X$ lies within
distance $\epsilon$ of~$P$, for some $\e<\min\{r(X)/16,\,
\rho(X)/16,\,s/8c\}$, where $s>0$ is the minimum
distance of the points of~$\Ext_1(f)$ to the diagonal $\Delta$.
Let $\delta\in\left[4\epsilon,\,\min\{r(X)/4,\,\rho(X)/4,\,s/2c\}\right)$, and $G_\delta(P)$
be the $\delta$-neighborhood graph built on top of $P$ with parameter $\delta$.
Then, the following inequality holds:
\[ 
d\left(\DM(X,f,\I), \, \DM(G_\delta(P),{\hat f}, \I))\right)\leq 2c\delta,
\]
where $\hat{f}$ is the piecewise linear interpolation of $f$ along the edges of $G_\delta(P)$ \cite{piecewiselinearinterpolation}.
\end{thm}
\begin{proof}
Let $C_X=\min \{\|x-x'\|_d : x,x'$ do not belong to the same connected component of $X\}$, 
and let $x,x'\in X$ be two points achieving $C_X$.
Let $y=\frac{1}{2}(x+x')\in\R^d$. Then $\|x-y\|_d \geq r(X)$
since $y$ belongs to the medial axis of $X$. Hence, $C_X = 2\|x-y\|_d \geq 2r(X)$.
Since $\delta < \frac{1}{4}r(X)<C_X$, it follows that $X$ and $\Rips^1_\delta(P)$ 
have the same number of connected components.
Then, the proof of Theorem~\ref{th:sig-approx} follows the same line
as the proof of Theorem~\ref{thm:perturb_bis}.
The only difference in the proof is the use of Theorem~7.5 in~\cite{Carriere15b} instead of the stability theorem~\cite{Cohen07}.
\end{proof}

\subsubsection{Hypothesis verification}

\label{verification}
In order to use those theorems, one needs to verify their hypothesis. Hence, the topology induced by the distance $d^*$ should verify that it is equivalent to the euclidean distance to be able to use the last theorem. Moreover, the function need to be Lipschitz in order to use the theorems \ref{thm:DStab_bis} and \ref{th:sig-approx}. Lastly, $f$ needs to be of Morse-type in order to use all of the theorems of stability (\ref{completeness},  \ref{thm:perturb_bis},\ref{thm:DStab_bis},\ref{th:sig-approx}).\\

For that, we will proceed in several steps : 
Let $x,y$ be two deviation components, whence $x,y \in \mathbb{R}^n.$ Then,
$d^*(x,y) = d_M(x,y) +  \bar{d_E} (x,y),$ where $\bar{d_E} (x,y)$ is the bounded euclidean distance by $1/4$ and $d_M(x,y)$ is given by $d_M(x,y) = \sum_{i=1}^n d_{m_i} (x_i,y_i),$ where 
\begin{equation}
\begin{small}
d_{m_i} (x_i,y_i) = \begin{cases} 0 &\mbox{if }  \textrm{sign}(x_i)=\textrm{sign}(y_i),\\ 
1 & \mbox{if }  \textrm{sign}(x_i) \neq\textrm{sign}(y_i) \\  
&\textrm{ and } |x_i| \textrm{ or }|y_i| \geq \alpha  \\
\frac{|x_i-y_i|}{8\alpha n} & \mbox{otherwise } \end{cases} 
\end{small}
\end{equation}
We observe that even if all the values are noise smaller than $\alpha$ (around 0), then the $d(x,y) < 1/2$, and therefore not perturbing the results if we replace $\epsilon$ by $\epsilon + 1/2$ in the corresponding section.
We will prove that with this distance
\begin{enumerate}
\item defines a topology.
\item verifies that $(\mathbb{R}^n, \mathcal{T}_{d^*}) = (\mathbb{R}^n,\mathcal{T}_{\bar{E}}),$ which is the topology with the bounded euclidean distance and which is known (Munkres) to be the same as $(\mathbb{R}^n,\mathcal{T}_{{E}})$, the standard topology with the euclidean distance.
\item  the
$$\begin{array}{ccc}
f:(\mathbb{R}^n, \mathcal{T}_{d^*})& \rightarrow & (\mathbb{R},\mathcal{T}_E) \\
x=(x_1,\dots, x_n) & \mapsto & \sum_{i=1}^n |x_i| 
\end{array}$$ is Lipschitz.
\item the function $f$ is Morse-type.

\end{enumerate} 
\begin{enumerate}

\item Let us show that $\{B_{d^*}(x,\epsilon) \mid \epsilon >0, x \in \mathbb{R}^n  \}$ defines a base of a topology. 
Indeed, for every $x \in \mathbb{R}^n$ and $x \in B_{d^*}(x,1/2)$ so the first axiom is verified. Secondly, let $x,y \in \mathbb{R}^n$ be two vectors and $\delta$ and $\epsilon$ two real numbers then, let $$t \in B_{d^*}(x,\delta) \cap B_{d^*}(y,\epsilon).$$ Let 
$$\begin{array}{cc}
\nu = & \min \Bigg\{ \min \{ |t_i | \mid t_i \neq 0, i = 1, \dots, n\}, \\
 & \min \{ |\alpha - |t_i| | \mid t_i \neq \alpha, i = 1, \dots, n\}, \\
& 1/4, \delta- d^*(x,t), \epsilon -d^*(y,t) \Bigg\}>0.
\end{array}
$$

We want to show that $B_{d^*}(t, \nu) \subseteq B_{d^*}(x,\delta) \cap B_{d^*}(y,\epsilon)$. The proof is the same for $y$ and $x$ just replacing $\epsilon$ by $\delta.$ Let us therefore focus on showing that $B_{d^*}(t, \nu) \subseteq B_{d^*}(x,\delta)$. \\

Let $z \in B_{d^*}(t, \nu)$. Hence, $d^*(z,t) < \nu$ and therefore
$$d_M(z,t) = \sum_{i=1}^n d_{m_i} (z_i,t_i) < \nu$$ since $d_{m_i} (z_i,t_i) \geq 0$ this means that $d_{m_i}(z_i,t_i) < \nu$ for every $i= 1, \dots, n$.\\

\newtheorem*{lem*}{Lemma}
\begin{lem*} [A]
If $t_i \neq 0$, then $z_i \neq 0$ and $\textrm{sign}(z_i)=\textrm{sign}(t_i)$.
\end{lem*} 
\begin{proof}
  Since, $d^*(z,t) \leq \nu$ then $d_{\bar{E}}(z,t) \leq \nu$.\\
 As $\nu < 1/4,$ $d_{\bar{E}}(z,t) = d_{E}(z,t)$  and hence 
$\sum_{i=1} |z_i -t_i| < \nu.$ 

This implies that $|z_i-t_i| < \nu$ for every $i = 1, \dots,n$.  \\
Moreover, since $|t_i| \neq 0$, we have that$ |z_i-t_i| < \nu \leq |t_i|$.
Therefore, $ z_i  \neq 0$ because otherwise we get $ |t_i| < |t_i|$, which is a contradiction. Moreover if $t_i >0$ and $z_i <0$ then $|z_i-t_i| = t_i - z_i = |t_i| + |z_i|,$ since $z_i$ is negative. This can not be strictly smaller than $|t_i|$ otherwise we get $|z_i | < 0$ which is a contradiction.  \\
Similarly if $t_i <0$ and $z_i >0,$ then $ |z_i-t_i| = z_i -t_i = |z_i| + |t_i|,$ which can not be strictly smaller than $|t_i|$.
Therefore, $ z_i $ and $t_i$ must have the same signature. \\
\end{proof}
\begin{lem*}[B]
If  $|t_i| \neq \alpha$ either $|t_i|$ and $|z_i|$ are $>\alpha$ or $|t_i|$ and $|z_i| < \alpha.$
\end{lem*}
\begin{proof}
 By the above argument, we know that $ |z_i-t_i| < \nu$ and $\nu \leq \mid \alpha - |t_i| \mid$. \\

Let us suppose $|t_i| > \alpha$ then $\mid \alpha - |t_i| \mid =  |t_i| - \alpha$, and 
$|z_i-t_i| < |t_i| - \alpha $ implies $-|z_i-t_i| > -|t_i| + \alpha$, which results in 
$$|z_i| \geq |t_i| -  |z_i-t_i| > |t_i| - |t_i| + \alpha = \alpha.$$ Hence, $|z_i| > \alpha$.\\

Let us suppose $|t_i| < \alpha$ then $\mid \alpha - |t_i| \mid = \alpha- |t_i|$, and
$$|z_i| \leq |z_i-t_i| + |t_i| <\alpha - |t_i| + |t_i| = \alpha$$ and hence $|z_i| < \alpha$.
\end{proof}

Let us enumerate the cases : 
\begin{itemize}
\item $H = \{ i \in \{1, \dots, n\} \mid  |t_i| = \alpha\}.$
\item $I = \{ i \in \{1, \dots, n\} \mid i \notin H, \textrm{ sign}(t_i)=0 \}.$
\item $J = \{ i \in \{1, \dots, n\} \mid  i \notin H\cup I, \textrm{ sign}(x_i)= \textrm{sign}(t_i) \}.$
\item $K = \{ i \in \{1, \dots, n\} \mid i \notin H\cup I, \textrm{ sign}(x_i) \neq \textrm{sign}(t_i), |t_i| < \alpha \}.$
\item $L = \{ i \in \{1, \dots, n\} \mid i \notin H\cup I, \textrm{ sign}(x_i) \neq \textrm{sign}(t_i), |t_i| \geq \alpha \}.$

\end{itemize}

Let us calculate, 
\begin{align*}
d_M(x,z) = &\sum_{i=1}^n d_{m_i}(x_i,z_i) = \sum_{i\in H}d_{m_i}(x_i,z_i)\\
& + \sum_{i\in I}d_{m_i}(x_i,z_i) 
  +  \sum_{i\in J} d_{m_i}(x_i,z_i) \\
& +  \sum_{i\in K } d_{m_i}(x_i,z_i)  +  \sum_{i\in L} d_{m_i}(x_i,z_i) 
\end{align*}
Now \begin{itemize}
\item For $i \in H$, there are two cases:
\begin{itemize}
\item If $d_{m_i}(x_i,t_i) = 0$, then by Lemma A, we have that $\textrm{sign}(t_i) = \textrm{sign}(z_i)$ and therefore $\textrm{sign}(z_i) = \textrm{sign}(x_i )$, and therefore $d_{m_i}(x_i,z_i) = d_{m_i}(x_i,t_i) = 0$.
\item If $d_{m_i}(x_i,t_i) = 1$, then $d_{m_i}(x_i,z_i) < d_{m_i}(x_i,t_i) = 1$.
\end{itemize}
\item For $i \in I $ since $t_i = 0$ there are several scenarios : \begin{itemize}
\item $|x_i | \geq \alpha$ : in this case either $z_i$ and $x_i$ have the same signature and then $d_{m_i} (x_i,z_i) = 0 < d_{m_i} (x_i, t_i)$ or the have opposite signatures and then $d_{m_i} (x_i,z_i) = 1 = d_{m_i} (x_i, t_i)$. In both cases $d_{m_i} (x_i,z_i)\leq  d_{m_i} (x_i, t_i)$.
\item $0<|x_i | < \alpha$ : if $\textrm{sign}(z_i)=\textrm{sign}(x_i)$, then $d_{m_i} (x_i,z_i) = 0 < d_{m_i}(x_i,t_i)$, otherwise by Lemma B as $t_i < \alpha$ we have that $z_i$ is smaller than $\alpha$ as well and hence $\textrm{sign}(z_i) \neq\textrm{sign}(x_i)$ then $d_{m_i} (x_i,z_i) =\frac{|x_i-z_i|}{8n\alpha} = \frac{|x_i|}{8n\alpha}+ \frac{|z_i|}{8n\alpha} =  d_{m_i} (x_i,t_i) +d_{m_i} (t_i,z_i)$.
\item $|x_i|= 0$ then $t_i= x_i$ and $d_{m_i}(x_i,z_i) = d_{m_i}(t_i,z_i)$
\end{itemize}
\item For $i \in J$  since $\textrm{sign}(t_i)=\textrm{sign}(x_i)$ and from Lemma A, we know that $\textrm{sign}(t_i) = \textrm{sign}(z_i).$ Therefore, $d_{m_i}(x_i,z_i) = d_{m_i} (x_i, t_i)= 0$.
\item  For $i \in K$, then $ |t_i| < \alpha$, and we know from Lemma B. that this implies $|z_i| < \alpha$ as well. We again have two cases here : 
\begin{itemize}
\item $|x_i| \geq \alpha$, $d_{m_i}(x_i,z_i) = 1 = d_{m_i}(x_i,t_i)$.\\
\item   $|x_i| < \alpha$, then $d_{m_i}(x_i,z_i) = \frac{|x_i-z_i|}{8n\alpha} \leq \frac{|x_i-t_i|}{8n\alpha} + \frac{|z_i-t_i|}{8n\alpha} = d_{m_i}(x_i,t_i)  + \frac{|t_i-z_i|}{8n\alpha} = d_{m_i}(x_i,t_i) + d_{m_i} (t_i,z_i) $.
\end{itemize}
\item For $i \in L$, since $ |t_i| \geq \alpha$, we know from Lemma B that $ |z_i| \geq \alpha$ as well, which implies that $d_{m_i} (x_i,t_i) = 1= d_{m_i} (x_i,z_i)$. 
\end{itemize}

Put together we have that 
\begin{align*}
d_M(x,z) =& \sum_{i=1}^n d_{m_i}(x_i,z_i) \\
=&  \sum_{i\in H}d_{m_i}(x_i,z_i) + \sum_{i\in I}d_{m_i}(x_i,z_i)  \\
&+  \sum_{i\in J} d_{m_i}(x_i,z_i)  +  \sum_{i\in K } d_{m_i}(x_i,z_i) \\
& +  \sum_{i\in L} d_{m_i}(x_i,z_i) \\
 \leq &
\sum_{i\in H}d_{m_i}(x_i,t_i)+ \sum_{i\in I}d_{m_i}(x_i,t_i) \\
&+  \sum_{i\in J} d_{m_i}(x_i,t_i) + d_{m_i}(t_i,z_i) \\
 & +  \sum_{i\in K } d_{m_i}(x_i,t_i) + d_{m_i}(t_i,z_i) \\
 & +  \sum_{i\in L} d_{m_i}(x_i,t_i) \\
  \leq &  d_M(x,t) + d_M(t,z) 
 \end{align*}
 
 Hence, 
 $d^*(x,z) = d_M(x,z) + d_{\bar{E}}(x,z) \leq d_M(x,t) + d_{\bar{E}}(x,t) + d_M(t,z) + d_{\bar{E}}(t,z)= d^*(x,t) + d^*(t,z) \leq d^*(x,t) + \delta-d^*(x,t) = \delta.$
\item $"\supseteq"$ Let $\epsilon >0$ and let $x \in \mathbb{R}^n$ if $\delta = \epsilon >0$ then $$B_{d^*}(x,\delta) \subseteq  B_{\bar{E}}(x,\epsilon) .$$ Indeed, if $y \in B_{d^*}(x,\delta) $, then $ d^*(x,y) < \delta$ and hence, $d_{\bar{E}}(x,y) \leq d_M(x,y) +  d_{\bar{E}}(x,y),$ since $d_M(x,y) \geq 0$ for every $x,y$ and hence $d_{\bar{E}}(x,y)  \leq d^*(x,y) < \delta = \epsilon.$ Therefore, $y \in B_{\bar{E}}(x,\epsilon)$. \\

 $"\subseteq"$ Let $\epsilon >0$ and let $x \in \mathbb{R}^n$ if $\delta = \min(\alpha/2,1/4, \epsilon/(\frac{1}{8\alpha}+1)) >0$ then $$B_{\bar{E}}(x,\delta) \subseteq  B_{d^*}(x,\epsilon) .$$ Indeed, if $y \in B_{\bar{E}}(x,\delta) $, then $ d_{\bar{E}}(x,y) < \delta$ and since $\delta < \alpha/2,$ and $\delta <1/4$, then for every $ i \in \{1, \dots, n\}$ either $\textrm{ sign}(x_i) = \textrm{ sign}(y_i)$ or $\textrm{ sign}(x_i) \neq \textrm{ sign}(y_i)$  and
  both $|x_i|$ and $|y_i|$ are less than or equal to $\alpha$.  \\
  We use the fact that $d_{\bar{E}}(x,y) < 1/4$ implies 
  $d_{\bar{E}}(x,y) = d_E(x,y)$. Then, $d_{\bar{E}}(x,y) < \alpha/2$ implies that $\sum_{i=1}^n |x_i-y_i| < \alpha/2$ 
  and therefore $|x_i-y_i| < \alpha/2.$  \\
  If $\textrm{sign}(x_i) \neq \textrm{sign}(y_i),$ then either $x_i >0$ and $y_i <0,$ 
  implying that $|x_i-y_i| = x_i-y_i > x_i = |x_i|$ and therefore $|x_i| < \alpha/2.$ This in turn implies that $ |y_i| < |x_i-y_i| + |x_i| < \alpha/2 + \alpha/2 = \alpha,$ or $y_i > 0$ and $x_i <0$ which with the same reasoning shows that $|x_i|$ and $|y_i|$ are smaller than $\alpha.$ \\

Coming back to the original problem, we obtain either $d^*(x,y) = d_{\bar{E}} (x,y)$ when  $\textrm{ sign}(x_i) = \textrm{ sign}(y_i)$ or  $$d^*(x,y) = d_M(x,y) +d_{\bar{E}}(x,y) = \sum_{i \in I} \frac{|x_i-y_i|}{8\alpha n}+ d_{\bar{E}}(x,y),$$ and since the $L^1$ norm is bounded by $\sqrt{n}$ times the $L^2$ norm, it is clear that $$\frac{|x_i-y_i|}{8\alpha n} \leq \frac{\sqrt{n}}{8\alpha n} \cdot d_{\bar{E}}(x,y)  \leq \frac{1}{8\alpha } \cdot d_{\bar{E}}(x,y) .$$ Therefore, $$d^*(x,y)  \leq  \frac{1}{8\alpha } \cdot  \delta + \delta \leq \epsilon,$$ where $I=\{i \in \{1, \dots,n \}\mid \textrm{sign}(x_i) \neq\textrm{sign}(y_i) \textrm{ and } |x_i| \textrm{ and }|y_i| \leq \alpha  \}.$ Therefore, $y \in B_{d^*}(x,\epsilon)$. 
\item $f$ is Lipschitz since $$f: (\mathbb{R}^n, \mathcal{T}_{d^*}) = (\mathbb{R}^n,\mathcal{T}_{{E}}) \rightarrow (\mathbb{R}, \mathcal{T}_{E}) .$$ and hence 

\begin{align*}
d(f(x),f(y)) &= \mid f(x)-f(y) \mid = \mid \sum_{i=1}^n |x_i| - \sum_{i=1} |y_i| \mid \\
&\leq  \sum_{i=1}^n \mid x_i - y_i \mid \leq d_E (x,y).
\end{align*} 
\item It is clearly of Morse-Type, since now $f:(\mathbb{R}^n, \mathcal{T}_{d^*}) = (\mathbb{R}^n,\mathcal{T}_{{E}}) \rightarrow (\mathbb{R}, \mathcal{T}_{E}) $ is the $L^1$-norm. Each interval in $\R$ has as pre-image a void thickened diamond in $\R^n$, which is compact and locally connected. Since the thickening is given by the length of the interval, it is then straightforward to obtain the needed homeomorphism and conclude that it is of Morse-type. 
\end{enumerate} 

 \bibliographystyle{unsrt}

\setcounter{figure}{0}
\beginsupplement

\begin{figure*}[h]
    \centering
    \renewcommand{\figurename}{Supplementary Figure}
        \includegraphics[height=18cm]{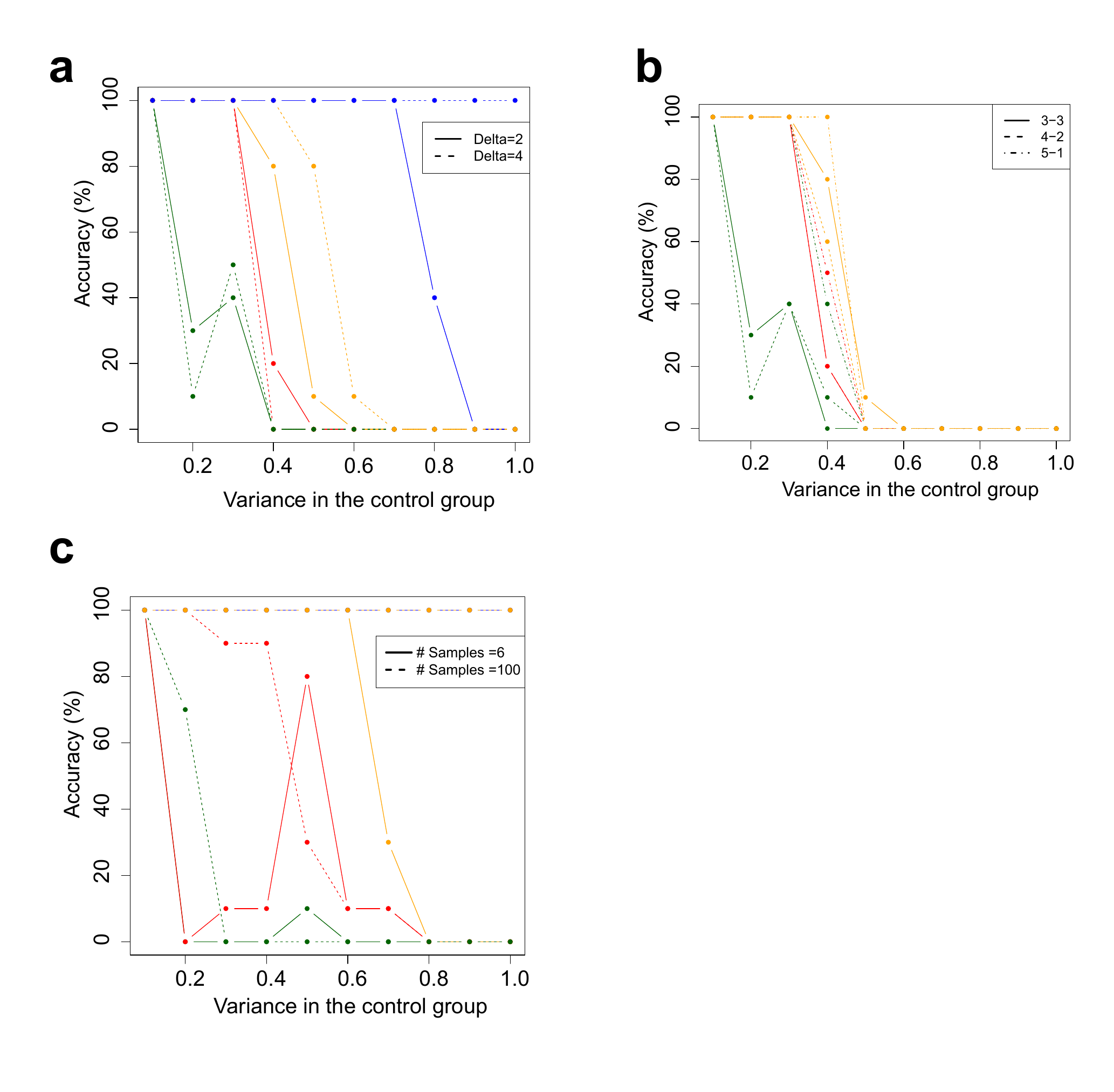}
	\caption{\label{SUP} Supplementary in silico data (a) accuracy plot when increasing delta b) accuracy when the subgroups have different sizes (c) accuracy plot when increasing the number of samples in the control group (d) accuracy plot when the sample size both in the control and in the test group are enlarged}
\end{figure*}

\newpage

\begin{table*}[h]
\begin{center}
%
%
\begin{tabular}{cccc}
\hline 
Name &
Abbr. & Name & Abbr.\\
\hline
Adult Accessory gland 	&	A	&	Adult Ovary 	&	O	\\
Adult Brain 	&	B	&	Larval Feeding Fatbody 	&	Fq	\\
Adult Carcass 	&	C&	Larval Feeding Carcass	&	Fc	\\
Adult Crop 	&	R	&	Adult Salivary Gland 	&	S	\\
Adult Heart 	&	D	&	Adult Spermatheca Mated 2	&	K3	\\
Adult Eye 	&	E	&	Larvae Wandering Tubules 	&	Wt	\\
Larval Feeding Hind Gut 	&	Fg	&	Adult Testes 	&	T\\
Adult Hind Gut	&	G	&	Adult Thoracic Muscle &	V		\\
Adult Head 	&	H	&	Adult Trachea & X	\\
Larval Feeding Mid Gut 	&	Fm	&	Adult Thoracoabdominal ganglion 	&	U	\\
Larval Feeding Salivary Gland 	&	Fs	&	Larval Feeding CNS	&	Fn	\\
Adult Spermatheca Mated 	&	K	&	Larval Wandering fat body	&	W	\\
Adult Spermatheca Virgin 	&	K2	&	Adult Wings	&	P	\\
Adult Mid Gut 	&	M	&	Whole Larvae Feeding 	&	F\\
Adult Ejaculatory Duct 	&	Z	&	Larval Feeding Trachea	&	Fx	\\
Larval Feeding Malpighian Tubule 	&	Ft	&	5th Passage Drosophila S2 Cells 	&	Y	\\
	&		&	Adult Fatbody 	&	Q	\\

\hline

\end{tabular}
\caption{Legend used for the fly data set, Abbr. = Abbreviation}\label{Tab:01} 
\end{center}
\end{table*}

\begin{figure*}[h]
    \centering
    \renewcommand{\figurename}{Supplementary Figure}
        \includegraphics[height=9cm]{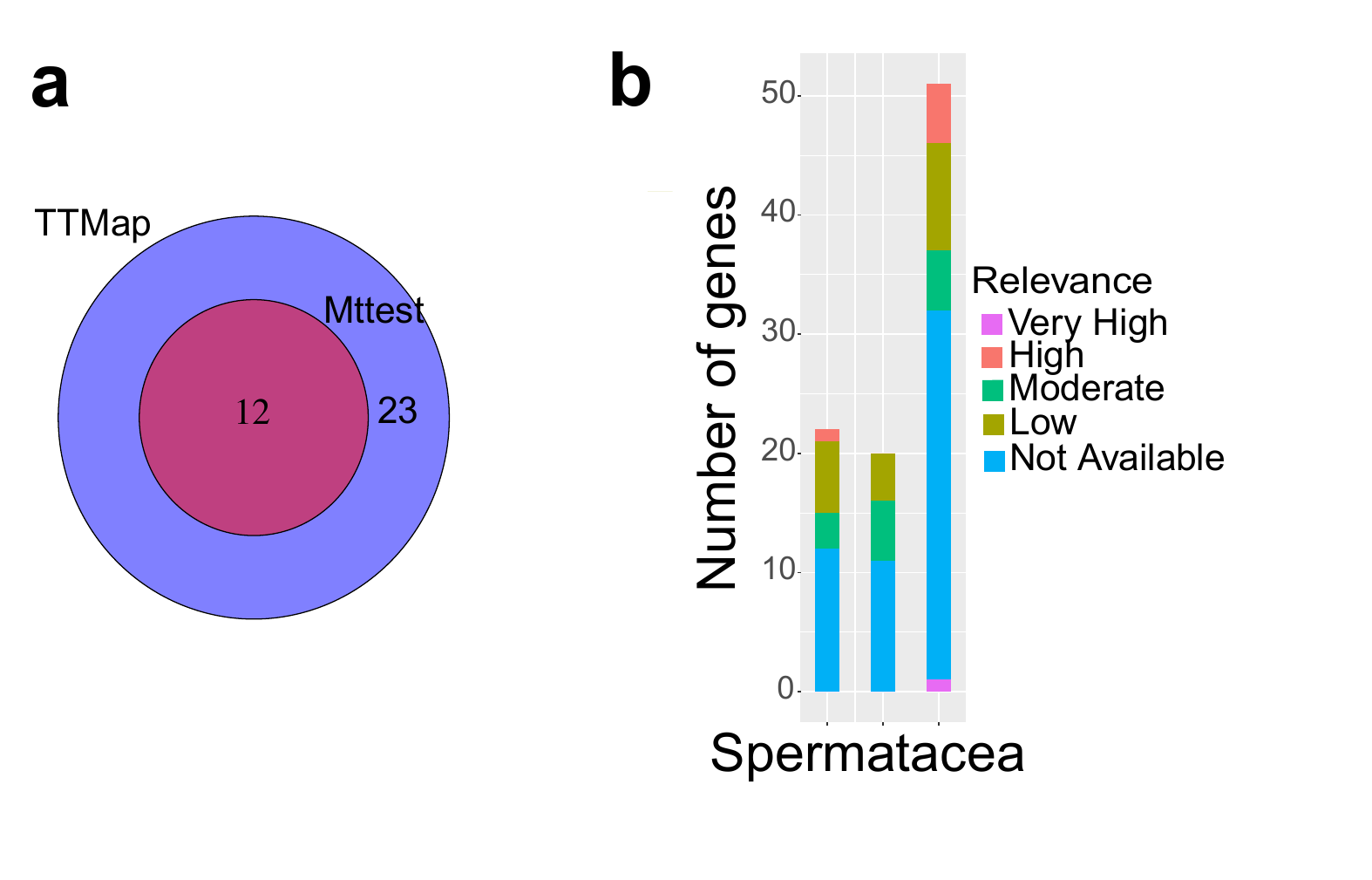}
	\caption{\label{fig_fly_sup} (a) Venn diagramm of significant genes of K with TTMap and with moderated-t-test (Mttest). (b) Barplot showing the relevance of the genes missed by Mttest on K, K2 and K3 }
\end{figure*}

\begin{figure*}[h]
    \centering
    \renewcommand{\figurename}{Supplementary Figure}
        \includegraphics[height=18cm]{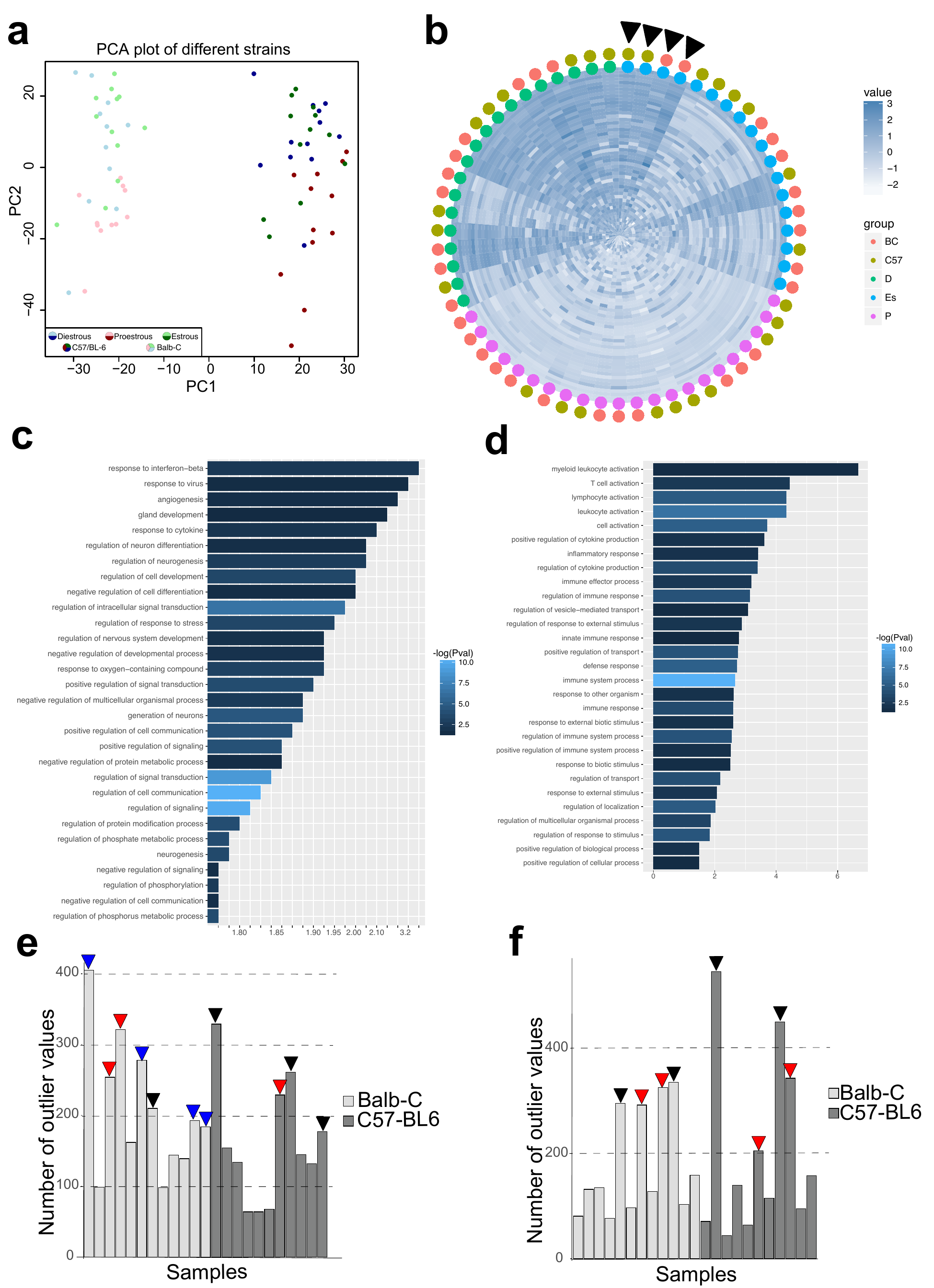}
	\caption{\label{supPCA} (a) PCA plot of RNA-seq profiles of mammary glands from Balb-C (light) and C57/BL6 mice (dark), in different phases of the estrous cycle. (b) Circle plot heatmap showing significant genes when the proestrous phase is considered as the control, without strain constraint. (c) heatmap of Panther pathway analysis by Fold Change with \textit{-log(Pval)} as a color code of E vs P  and (d) of E vs D }
\end{figure*}

\begin{figure*}[h]
    \centering
    \renewcommand{\figurename}{Supplementary Figure}
        \includegraphics[height=17cm]{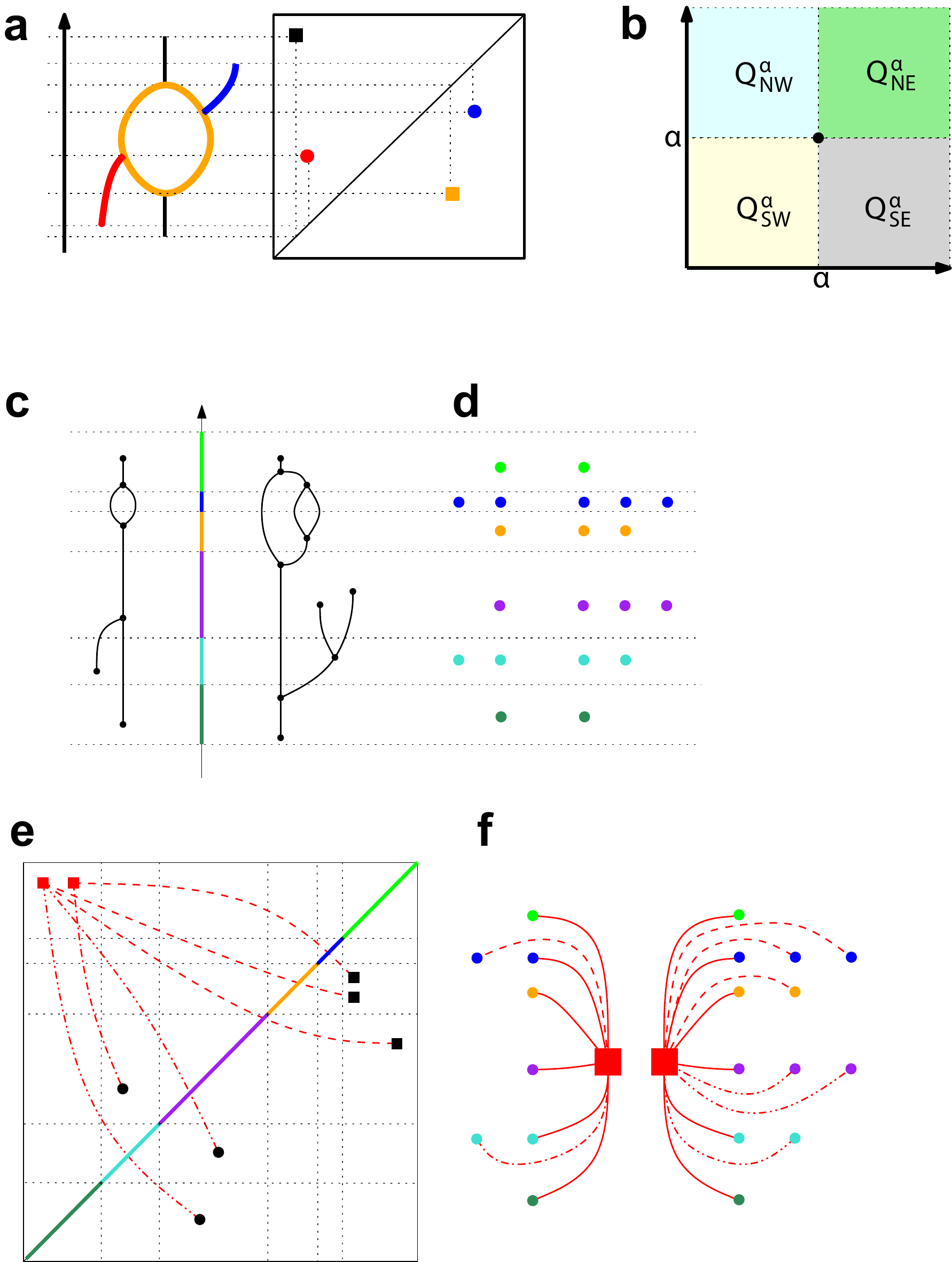}
	\caption{\label{fig:sign} (a) Example of correspondences between topological features of a graph
	and points in its corresponding extended persistence diagram. Note that ordinary persistence
	is unable to detect the blue upwards branch. (b) Plot of the various $Q_*^\alpha$ in the plane. (c) The Reeb graph and (d) its Mapper computed with a cover of $\im(f)$ with disjoint intervals.
	(e) By adding $\R$ to this cover, the descriptor is calculated and (f) the Mapper can be retrieved from the descriptor in (e).}
\end{figure*}
\end{document}